\documentclass[11pt,oneside]{article}

\usepackage{amsthm, amsmath, amssymb, amsfonts, graphicx, epsfig, mathtools, mathdots}

\usepackage{bbm}

\usepackage{enumerate}

\usepackage[normalem]{ulem}

\usepackage[colorlinks=true, pdfstartview=FitV, linkcolor=blue,
            citecolor=blue, urlcolor=blue]{hyperref}
\usepackage[usenames]{color}
\definecolor{Red}{rgb}{0.7,0,0.1}
\definecolor{Green}{rgb}{0,0.7,0}

\usepackage{url}
\makeatletter\def\url@leostyle{%
 \@ifundefined{selectfont}{\def\UrlFont{\sf}}{\def\UrlFont{\scriptsize\ttfamily}}} \makeatother

\newtheorem{theorem}{Theorem}[section]
\newtheorem{proposition}[theorem]{Proposition}
\newtheorem{lemma}{Lemma}[section]

\theoremstyle{definition}
\newtheorem{definition}{Definition}[section]
\newtheorem{example}{Example}[section]
\theoremstyle{remark}
\newtheorem{remark}[theorem]{Remark}

\newtheoremstyle{dotless}{}{}{\itshape}{}{\bfseries}{}{ }{}
\theoremstyle{dotless}

\def\cC{\mathcal{C}}

\def\cF{\mathcal{F}}

\def\cI{\mathcal{I}}
\def\cJ{\mathcal{J}}
\def\cK{\mathcal{K}}

\def\cP{\mathcal{P}}

\def\cR{\mathcal{R}}
\def\cS{\mathcal{S}}
\def\cT{\mathcal{T}}

\def\cW{\mathcal{W}}

\def\cZ{\mathcal{Z}}

\def\bE{\mathbb{E}}
\def\bF{\mathbb{F}}

\def\bK{\mathbb{K}}

\def\bN{\mathbb{N}}

\def\bP{\mathbb{P}}
\def\bQ{\mathbb{Q}}
\def\bR{\mathbb{R}}

\def\1{\mathbbm{1}}

\DeclareMathOperator*{\esssup}{ess\,sup} 
\DeclareMathOperator*{\essinf}{ess\,inf} 

\def\d{\mathrm{d}}

\def\esssup{\text{ess sup}}
\def\essinf{\text{ess inf}}

\title{No-Arbitrage Pricing for Dividend-Paying Securities in Discrete-Time Markets with Transaction Costs}
\author{Tomasz R. Bielecki\footnote{Tomasz R. Bielecki and Igor Cialenco acknowledge support from the NSF grant DMS-0908099.}\\
\small{Department of Applied Mathematics,}\\[-0.3ex]
\small{Illinois Institute of Technology,}\\[-0.3ex]
\small{Chicago, 60616 IL, USA}\\[-0.3ex]
\url{bielecki@iit.edu}\\
\and
Igor Cialenco\footnotemark[\value{footnote}]\\[-0.3ex]
\small{Department of Applied Mathematics,}\\[-0.3ex]
\small{Illinois Institute of Technology,}\\[-0.3ex]
\small{Chicago, 60616 IL, USA}\\[-0.3ex]
\url{igor@math.iit.edu} \\
\and
Rodrigo Rodriguez \\
\small{Department of Applied Mathematics,}\\[-0.3ex]
\small{Illinois Institute of Technology,}\\[-0.3ex]
\small{Chicago, 60616 IL, USA}\\[-0.3ex]
\url{rrodrig8@hawk.iit.edu}}

\date{First Circulated: May 25, 2012\\
This Version: June 02, 2013, Extended Preprint \\
Forthcoming in Mathematical Finance}

\begin{document}
\maketitle

\begin{abstract}
We prove a version of First Fundamental Theorem of Asset Pricing under transaction costs for discrete-time markets with dividend-paying securities.
Specifically, we show that the no-arbitrage condition under the efficient friction assumption is equivalent to the existence of a risk-neutral measure.
We derive dual representations for the  superhedging ask and subhedging bid price processes of a contingent claim.
Our results are illustrated with a vanilla credit default swap contract.
 \end{abstract}

{\noindent \small
}


\newpage
\section{Introduction}

One of the central themes in mathematical finance is no-arbitrage pricing and its applications.
At the foundation of no-arbitrage pricing is the First Fundamental Theorem of Asset Pricing (FFTAP).
We prove a version of the FFTAP under transaction costs\footnote{In this study, a transaction cost is defined as the cost incurred in trading in a market in which securities' quoted prices have a bid-ask spread. We do not consider other costs  such as broker's fees and taxes.} for discrete-time markets with dividend-paying securities.

The FFTAP has been proved in varying levels of generality for frictionless markets.
In a discrete-time setting for a finite state space, the theorem was first proved in Harrison and Pliska \cite{HarrisonPliska1981}.
Almost a decade later, Dalang, Morton, and Willinger \cite{DalangMortonWillinger1990} proved the FFTAP for the more technically challenging setting in which the state space is general.
Their approach requires the use of advanced, measurable selection arguments, which motivated several authors to provide alternative proofs using more accessible techniques (see Schachermayer~\cite{Schachermayer1992}, Kabanov and Kramkov~\cite{KabanovKramkov1994}, Rogers~\cite{Rogers1994}, Jacod and Shiryaev~\cite{JacodShiryaev1998}, and Kabanov and Stricker~\cite{KabanovStricker2001}).
Using advanced concepts from functional and stochastic analysis, the FFTAP was first proved in a general continuous-time set-up in the celebrated paper by Delbaen and Schachermayer~\cite{DelbaenSchachermayer1994}.
A comprehensive review of the literature pertaining to no-arbitrage pricing theory in frictionless markets can be found in Delbaen and Schachermayer~\cite{DelbaenSchachermayer2006}.

The first rigorous study of the FFTAP for markets with transaction costs in a discrete-time setting was carried out by Kabanov and Stricker~\cite{KabanovStricker}.
Under the assumption that the state space is finite, it was proved that \textbf{NA} is equivalent to the existence of a consistent pricing system (using the terminology introduced in Schachermayer~\cite{Schachermayer2004}).
However, their results did not extend to the case of a  general state space.
As in the frictionless case, the transition from a finite state space to a general state space is nontrivial due to measure-theoretic and topological related difficulties.
These difficulties were overcome in  Kabanov, R\'{a}sonyi, and Stricker~\cite{KabanovRasonyiStricker2002}, where a version of the FFTAP was proven under the \emph{efficient friction assumption} (\textbf{EF}).
It was shown that the \emph{strict no-arbitrage condition}, a condition which is stronger than \textbf{NA}, is equivalent to the existence of a strictly consistent pricing system.
Therein, it was asked whether \textbf{EF} can be discarded.
Schachermayer~\cite{Schachermayer2004} answered this question negatively by showing that neither \textbf{NA} nor the strict no-arbitrage condition alone is sufficiently strong to yield the existence of a consistent pricing system.
More importantly, Schachermayer~\cite{Schachermayer2004} proved a new version of the FFTAP that does not  require \textbf{EF}.
Specifically, he proved that the \emph{robust no-arbitrage condition}, which is stronger than the strict no-arbitrage condition, is equivalent to the existence of a strictly consistent pricing system.
Subsequent studies that treat the robust no-arbitrage condition are  Bouchard~\cite{Bouchard2006},  Valli\`{e}re, Kabanov, and Stricker~\cite{ValliereKabanovStricker2007}, Jacka, Berkaoui, and Warren \cite{JackaBerkaouiWarren2008}.
 Recently, Pennanen~\cite{Pennanen2011, Pennanen2011b,Pennanen2011c,Pennanen2011d} studied  no-arbitrage pricing  in a general context in which markets can have constraints and transaction costs may depend nonlinearly on traded amounts.
Therein, the problem of superhedging a claims process (e.g. swaps) is also investigated.
An excellent survey of the literature pertaining to no-arbitrage pricing in markets with transaction costs can be found in Kabanov and Safarian~\cite{KabanovSafarian2009}.
Let us mention that versions of the FFTAP for markets with transaction costs in a continuous-time setting have also been studied in the literature.
This literature considers stronger conditions than \textbf{NA} (see for instance Jouini and Kallal~\cite{JouiniKallal1995}, Cherny~\cite{Cherny2007a}, Guasoni, R\'{a}sonyi, and Schachermayer~\cite{GuasoniSchachermayer},  Denis, Guasoni, and R\'{a}sonyi~\cite{DenisGuasoniRasonyi2011}, Denis and Kabanov~\cite{DenisKabanov2012}).

The  fundamental difference between no-arbitrage pricing theory for dividend-paying securities and  non-dividend paying securities is that transaction costs associated with trading dividend-paying securities is that transaction costs associated with trading dividend-paying securities may not be proportional to the number of units of securities purchased or sold.
Transaction costs associated with dividend-paying securities may \emph{accrue} over time by merely holding the security---for a non-dividend paying security transaction costs are only charged whenever the security is bought or sold.
Our consideration of transaction costs on dividends distinguishes this study.

The contribution of this paper is summarized as follows:
\begin{itemize}
\item
We define and study the value process and the self-financing condition under transaction costs for discrete-time markets with dividend-paying securities (Section~\ref{section:ValueSF}).
\item
We define and investigate \textbf{NA} and \textbf{EF} in our context (Section~\ref{section:Arbitrage}).
\item
We prove a key closedness property of the set of claims that can be superhedged at zero cost (Section~\ref{subsection:Closedness}).
\item
Using classic separation arguments, we prove a version of the FFTAP that is relevant to our set-up.
Specifically, we prove that \textbf{NA} under \textbf{EF} is satisfied if and only if there exists a risk-neutral measure (Section~\ref{subsection:FFTAP}).
\item
We introduce an appropriate notion of consistent pricing systems in our set-up, and we study the relationship between them and \textbf{NA} under \textbf{EF} (Section~\ref{section:CPS}).
We demonstrate that, if there are no transaction costs on the dividends paid by securities, \textbf{NA} under \textbf{EF} is equivalent to the existence of a consistent pricing system (Section~\ref{subsection:CPSNoT}).
\item
We derive a dual representation for the superhedging ask and subhedging bid price processes for a contingent claim (Section~\ref{section:SuperHedging}).
\end{itemize}

\section{The value process and the self-financing condition}\label{section:ValueSF}
Let $T$ be a fixed time horizon, and let $\cT:=\{0, 1, \dots, T\}$ and $\cT^*:=\{1, 2, \dots, T\}$.
Next, let $(\Omega, \cF_T, \bF=(\cF_t)_{t\in \cT}, \bP)$ be the underlying  filtered probability space.

On this probability space, we consider a market consisting of a savings account $B$ and of $N$ traded securities satisfying the following properties:
\begin{itemize}
\item
The savings account can be purchased and sold according to the  price process $B:=\Big(\big(\prod_{s=0}^t(1+r_s)\big)\Big)_{t=0}^T$, where $(r_t)_{t=0}^T$ is  a nonnegative process specifying the risk-free rate.
\item
 The $N$ securities can be purchased according to the ex-dividend price process $P^{ask}:=\big((P^{ask,1}_t, \dots, P^{ask, N}_t)\big)_{t=0}^T$, and pay (cumulative) dividends specified by the process $A^{ask}:=\big((A^{ask, 1}_t, \dots, A^{ask, N}_t)\big)_{t=1}^T$.
 The quantity $\Delta A^{ask}_t$ is the dividends per unit of securities held long.
\item
The $N$ securities  can be sold  according to the ex-dividend price process $P^{bid}:=\big((P^{bid,1}_t, \dots, P^{bid, N}_t)\big)_{t =0}^T$, and pay (cumulative) dividends specified by the process $A^{bid}:=\big((A^{bid, 1}_t, \dots, A^{bid, N}_t)\big)_{t=1}^T$.
The quantity $\Delta A^{bid}_t$ is the dividends per unit of securities held short.
\end{itemize}
\noindent
We assume that the processes introduced above adapted.
In what follows, we shall denote by $\Delta $ the backward difference operator: $\Delta X_t:=X_t-X_{t-1}$, and we take the convention that $A^{ask}_0=A^{bid}_0=0$.
It is easy to verify the following product rule for $\Delta$:
    \begin{align*}
      \Delta (X_t Y_t)&=X_{t-1}\Delta Y_t+Y_t \Delta X_t=X_t\Delta Y_t+Y_{t-1}\Delta X_t.
    \end{align*}

\begin{remark}
For any $t =1, 2, \dots, T$ and $j=1, 2, \dots, N$, the random variable $\Delta A^{ask, j}_t$ is interpreted as amount of dividend associated with holding a \emph{long} position in security $j$ from time $t-1$ to time $t$, and the random variable $\Delta A^{bid, j}_t$ is interpreted as amount of dividend associated with holding a \emph{short} position in security $j$ from time $t-1$ to time $t$.
\end{remark}

We now illustrate the processes introduced above in the context of a vanilla  Credit Default Swap  (CDS) contract.
\begin{example}\label{example:CDS}
A CDS contract is a contract between two parties, a \emph{protection buyer} and a \emph{protection seller}, in which the protection buyer pays periodic fees to the protection seller in exchange for some payment made by the protection seller to the protection buyer if a pre-specified credit event of a reference entity occurs.
Let $\tau$ be the nonnegative random variable specifying the time of the credit event of the reference entity.
Suppose the CDS contract admits the following specifications: initiation date $t=0$, expiration date $t=T$, and nominal value \$1.
 For simplicity, we assume that the  loss-given-default is a nonnegative scalar $\delta$ and is paid at default.
Typically, CDS contracts are traded on over-the-counter markets in which dealers quote CDS spreads to investors.
Suppose that the CDS spread quoted by the dealer to sell a CDS contract with above specifications is $\kappa^{bid}$ (to be received every unit of time), and the CDS spread quoted by the dealer to buy a CDS contract with above specifications is $\kappa^{ask}$ (to be paid every unit of time).
We remark that the CDS spreads $\kappa^{ask}$ and $\kappa^{bid}$ are specified in the CDS contract, so the CDS contract to sell protection is technically a different contract than the CDS contract to buy protection.

The cumulative dividend processes $A^{ask}$ and $A^{bid}$ associated with buying and selling the CDS with specifications above, respectively, are defined as
\begin{align*}
 A^{ask}_t&:=1_{\{\tau \leq t\}} \delta -\kappa^{ask} \sum_{u=1}^t 1_{\{u <\tau\}}, \quad A^{bid}_t:=1_{\{\tau \leq t\}} \delta -\kappa^{bid} \sum_{u=1}^t 1_{\{u <\tau\}}
\end{align*}
for $t \in \cT^*$.
In this case, the ex-dividend ask and bid price processes $P^{bid}$ and $P^{ask}$ specify the mark-to-market values of the CDS for the protection seller and protection buyer, respectively, from the perspective of the protection buyer.
The CDS spreads $\kappa^{ask}$ and $\kappa^{bid}$ are set so that $P^{bid}_0=P^{ask}_0=0$.
Also, we have that $P^{ask}_T=P^{bid}_T=0$ since they are ex-dividend prices.
\end{example}

Next, we illustrate the processes above with a vanilla Interest Rate Swap (IRS) contract.

\begin{example}\label{example:IRS}
An IRS contract is a contract between two parties, in which one party agrees to periodically pays a fixed rate (the swap rate) to the other party, in exchange for a floating rate (usually the Libor rate).
We suppose that the floating rate from $i-1$ to $i$, denoted by $L_i$, is exchanged for the swap rate every unit of time.
Also, we assume that the IRS admits the following specifications: initiation date $t=0$, expiration date $t=T$, and nominal value \$1.
IRS contracts are traded on over-the-counter markets in which dealers quote swap rates to investors.
For the contract specified above, we denote by $s^{ask}$ the swap rate quoted by the dealer for a Payer IRS (pays the swap rate and receives the floating rate), and denote by $s^{bid}$ the swap rate quoted by the dealer for a Receiver IRS (pays the floating rate and receives the swap rate).
We remark that the spreads $s^{ask}$ and $s^{bid}$ are specified in the IRS contract.

The cumulative dividend processes $A^{bid}$ and $A^{ask}$ associated with the Payer and Receiver swap with specifications above, respectively, are defined as
\begin{align*}
 A^{ask}_t&:= \sum_{i=1}^t (L_i-s^{ask}), \quad A^{bid}_t:= \sum_{i=1}^t (s^{bid}-L_i)
\end{align*}
for $t \in \cT^*$.
The ex-dividend ask and bid price processes $P^{bid}$ and $P^{ask}$ specify the mark-to-market values of the IRS for the Payer IRS and Receiver IRS, respectively.
The values of swap spreads $s^{ask}$ and $s^{bid}$ are set so that $P^{bid}_0=P^{ask}_0=0$ are null at initiation date, and also note that $P^{ask}_T=P^{bid}_T=0$ since they are ex-dividend prices.
\end{example}

From now on, we make the following standing assumption.
\medskip

\noindent {\bf Bid-Ask Assumption:}
$\ P^{ask} \geq P^{bid}$ \textup{and} $\Delta  A^{ask} \leq \Delta  A^{bid}$.

\medskip

For convenience, we define $\mathcal{J}:=\{0, 1, \dots, N\}$ and\\
$\mathcal{J}^*:=\{1, 2, \dots, N\}$.
Unless stated otherwise, all inequalities and equalities between processes and random variables are understood $\bP$-a.s. and coordinate-wise.

\subsection{The value process and self-financing condition}

A \emph{trading strategy} is a predictable process $\phi:=\big((\phi^0_t, \phi^1_t, \dots, \phi^N_t)\big)_{t=1}^T$, where $\phi^j_t$ is interpreted as the number of units of security $j$ held from time $t-1$ to time $t$.
Processes $\phi^1, \dots, \phi^N$ correspond to the holdings in the  $N$ securities, and process $\phi^0$ corresponds to the holdings in the savings account $B$.
We take the convention $\phi_0=(0,\ldots,0)$.

\begin{definition}\label{Intro/DefValueProcess}
The \emph{value process} $(V_t(\phi))_{t=0}^T$ associated with a trading strategy $\phi$ is defined as
\begin{equation*}
V_t(\phi) =
\begin{cases}
\phi^{0}_1+\sum_{j=1}^N\phi^{j}_1( 1_{\{\phi^j_1 \geq 0\}}P^{ask, j}_0 +1_{\{\phi^j_1 < 0\}}P^{bid, j}_0), &\mbox{if $t=0$},\\[.15in]
\phi^{0}_tB_t+\sum_{j=1}^N \phi^{j}_t (1_{\{\phi^j_t \geq 0\}}P^{bid, j}_t+1_{\{\phi^j_t < 0\}}P^{ask, j}_t)\\[.02in]
 \qquad +\sum_{j=1}^N  \phi^{ j}_t(1_{\{\phi^j_t < 0\}}\Delta A^{bid,j }_t+1_{\{\phi^j_t \geq 0\}}\Delta A^{ask,j }_t),  &\mbox{if $1 \leq t\leq T$}.
  \end{cases}
  \end{equation*}
\end{definition}
For $t=0$, $V_0(\phi)$ is interpreted as the cost of the portfolio $\phi$, and for $t \in \{1, \dots, T\}$ it is interpreted as the liquidation value of the portfolio before any time $t$ transactions, including any dividends acquired from time $t-1$ to time $t$.

\begin{remark}\mbox{}
 Also note that, due to the presence of transaction costs, the value process $V$ may not be linear in its argument, i.e. $V_t(\phi)+V_t(\psi)\neq V_t(\phi+\psi)$, and $V_t(\alpha\phi)\neq \alpha V_t(\phi)$ for $\alpha \in\bR$,  and some trading strategies $\phi, \psi$, some time $t \in \cT$.
This is the major difference from the frictionless setting.
\end{remark}

Next, we introduce the self-financing condition, which is appropriate in the context of this paper.

\begin{definition}\label{Intro/DefSelfFinancing}
A trading strategy $\phi$ is self-financing if
\begin{align}
&B_{t}\Delta  \phi^{0}_{t+1}+\sum_{j=1}^N\Delta  \phi^{ j}_{t+1}(1_{\{\Delta  \phi^j_{t+1} \geq 0\}}P^{ask, j}_{t}+1_{\{\Delta  \phi^j_{t+1} < 0\}}P^{bid, j}_{t} ) \notag \\
&  \qquad \qquad =  \sum_{j=1}^N \phi^{ j}_{t}(1_{\{ \phi^j_{t} \geq 0\}} \Delta  A^{ask, j}_{t}+1_{\{ \phi^j_{t} < 0\}} \Delta  A^{bid, j}_{t}) \label{eq:SF}
  \end{align}
  for $t=1, 2, \dots, T-1$.
\end{definition}
The self-financing condition imposes the restriction that no money can flow in or out of the portfolio.
We note that if $P:=P^{ask}=P^{bid}$ and $\Delta A:=\Delta A^{ask} = \Delta A^{bid}$, then the self-financing condition in the frictionless case is recovered.

\begin{remark}
Note that the self-financing condition not only takes into account transaction costs due purchases and sales of securities (left hand side of \eqref{eq:SF}), but also transaction costs accrued through the dividends (right hand side of \eqref{eq:SF}).
\end{remark}

The next result gives a useful characterization of the self-financing condition in terms of the value process.

\begin{proposition}\label{Intro/LemmaSelfFinancing}
A trading strategy $\phi$ is self-financing if and only if the value process $V(\phi)$ satisfies
\begin{align}
V_t(\phi)&=V_0(\phi)+\sum_{u=1}^t\phi^{0}_u \Delta  B_u+\sum_{j=1}^N \phi^{j}_t \Big(1_{\{\phi^j_t \geq 0\}} P^{bid, j}_t +1_{\{\phi^j_t < 0\}}P^{ask, j}_t\Big ) \notag\\
& \qquad- \sum_{j=1}^N\sum_{u=1}^t\Delta \phi^{ j}_{u}\Big(1_{\{\Delta  \phi^j_u \geq 0\}}P^{ask, j}_{u-1}+1_{\{\Delta  \phi^j_u < 0\}}P^{bid, j}_{u-1}\Big) \notag\\
& \qquad +\sum_{j=1}^N\sum_{u=1}^t \phi^{ j}_u \Big(1_{\{\phi^j_u \geq 0\}}\Delta  A^{ask,j}_u+1_{\{\phi^j_u < 0\}}\Delta A^{bid,j}_u\Big) \label{eq:SFmain}
\end{align}
for all $t \in \mathcal{T}^*$.
\end{proposition}

\begin{proof}
By the definition of $V(\phi)$, and applying the product rule for the backwards difference operator $\Delta$, we obtain
\begin{align}\label{eq:SF2}
V_t(\phi)&=V_0(\phi)+\sum_{u=1}^t \phi^{0}_u \Delta  B_u+\sum_{u=2}^t B_{u-1}\Delta  \phi^{0}_u\\ \notag
& -\sum_{j=1}^N \phi^{j}_1\Big(1_{\{\phi^j_1 < 0\}}P^{bid, j}_0 + 1_{\{\phi^j_1 \geq 0\}}P^{ask, j}_0\Big)
+\sum_{j=1}^N \phi^{j}_t \Big(1_{\{\phi^j_t \geq 0\}} P^{bid, j}_t+ 1_{\{\phi^j_t < 0\}} P^{ask, j}_t\Big)\\ \notag
&  +\sum_{j=1}^N\sum_{u=1}^t \phi^{ j}_u \Big(1_{\{\phi^j_u \geq 0\}} \Delta A^{ask,j}_u+ 1_{\{\phi^j_u < 0\}} \Delta A^{bid,j}_u\Big) \notag\\
&-\sum_{j=1}^N\sum_{u=2}^t\phi^{ j}_{u-1}\Big(1_{\{\phi^j_{u-1} < 0\}}\Delta A^{bid,j}_{u-1}+1_{\{\phi^j_{u-1} \geq 0\}} \Delta A^{ask,j}_{u-1}\Big) \notag.
\end{align}

If $\phi$ is self-financing, then we see that \eqref{eq:SF2} reduces to \eqref{eq:SFmain}.
Conversely, assume that the value process satisfies \eqref{eq:SFmain}.
Subtracting \eqref{eq:SFmain} from \eqref{eq:SF2} and applying the product rule for the backwards difference $\Delta$ to both sides yields that $\phi$ is self-financing.
\end{proof}

The next proposition extends the previous result in terms of our \emph{num\'{e}raire} $B$.
For convenience, we let $V^*(\phi):=B^{-1}V(\phi)$ for all trading strategies $\phi$.

\begin{proposition}\label{Intro/LemmaDiscountedSelfFinancing}
A trading strategy $\phi$ is self-financing if and only if the discounted value process $V^*(\phi)$ satisfies
\begin{align}\label{eq:ValueSF}
V^*_t(\phi)&=V_0(\phi)+\sum_{j=1}^N \phi^{j}_t B^{-1}_t\Big(1_{\{\phi^j_t \geq 0\}}P^{bid, j}_t
+1_{\{\phi^j_t < 0\}}P^{ask, j}_t\Big)\notag\\
& \;  - \sum_{j=1}^N\sum_{u=1}^t\Delta\phi^{ j}_{u}B^{-1}_{u-1}\Big(1_{\{\Delta  \phi^j_u \geq 0\}}P^{ask, j}_{u-1}
+1_{\{\Delta  \phi^j_u < 0\}}P^{bid, j}_{u-1}\Big) \notag\\
& \; +\sum_{j=1}^N \sum_{u=1}^t\phi^{ j}_u B^{-1}_u\Big(1_{\{\phi^j_u \geq 0\}}\Delta A^{ask,j}_u
+1_{\{\phi^j_u < 0\}}\Delta A^{bid,j}_u\Big)
\end{align}
for all $t \in \mathcal{T}^*$.
\end{proposition}

\begin{proof}
Suppose that $\phi$ is self-financing.
We may apply Proposition~\ref{Intro/LemmaSelfFinancing} and the product rule for the backwards difference $\Delta$ to see that
\begin{align*}
\Delta\big(B^{-1}_t V_t(\phi)\big)&=\sum_{j=1}^N \Delta \Bigg(\phi^{j}_t B^{-1}_t\Big(1_{\{\phi^j_t \geq 0\}}P^{bid, j}_t
+ 1_{\{\phi^j_t < 0\}}P^{ask, j}_t\Big)\Bigg)\\
& \;  - \sum_{j=1}^N\Delta\phi^{ j}_{t}B^{-1}_{t-1}\Big(1_{\{\Delta  \phi^j_t \geq 0\}}P^{ask, j}_{t-1}
+1_{\{\Delta  \phi^j_t < 0\}}P^{bid, j}_{t-1}\Big)\\
& \; +\sum_{j=1}^N \phi^{ j}_t B^{-1}_t \Big(1_{\{\phi^j_t \geq 0\}}\Delta  A^{ask,j}_t+1_{\{\phi^j_t < 0\}}\Delta A^{bid,j}_t\Big)
\end{align*}
for all $t \in \cT^*$.
Summing both sides of the equation from $u=1$ to $u=t$ shows that necessity holds.

Conversely, if the value process $V(\phi)$ satisfies \eqref{eq:ValueSF}, we may apply the product rule for the backwards difference $\Delta$ to $\Delta\big(B(B^{-1}V(\phi))\big)$ to deduce that
\begin{align*}
\Delta V_t(\phi)&=\phi^{0}_{t}\Delta  B_t+\sum_{j=1}^N\Delta \Bigg( \phi^{j}_t\Big(1_{\{\phi^j_t \geq 0\}}P^{bid, j}_t+1_{\{\phi^j_t < 0\}}P^{ask, j}_t\Big)\Bigg)\\
& \;  - \sum_{j=1}^N\Delta\phi^{ j}_{t}\Big(1_{\{\Delta  \phi^j_t \geq 0\}}P^{ask, j}_{t-1}+1_{\{\Delta  \phi^j_t < 0\}}P^{bid, j}_{t-1}\Big)\\
& \; +\sum_{j=1}^N\phi^{ j}_t\Big(1_{\{\phi^j_t \geq 0\}}\Delta  A^{ask,j}_t+1_{\{\phi^j_t < 0\}}\Delta A^{bid,j}_t\Big).
\end{align*}
After summing both sides of the equation above from $u=1$ to $u=t$ and applying Proposition~\ref{Intro/LemmaSelfFinancing}, we see that $\phi$ is self-financing.
\end{proof}

\begin{remark}
If $P=P^{ask}=P^{bid}$ and $\Delta A=\Delta A^{ask}=\Delta A^{bid}$, then we recover the classic result: a trading strategy $\phi$ is self-financing if and only if the value process satisfies
\begin{equation*}
V^*_t(\phi)=V_0(\phi)+\sum_{j=1}^N\sum_{u=1}^t\phi^j_u\Delta \Big(B^{-1}_uP^j_u + \sum_{w=1}^u B^{-1}_w \Delta A^j_w\Big)
\end{equation*}
for all $t \in \cT^*$.
\end{remark}

For convenience, we define $P^{ask, *}:=B^{-1}P^{ask}, P^{bid, *}:=B^{-1}P^{bid}, A^{ask, *}:=B^{-1}\Delta A^{ask},$ and $A^{bid, *}:=B^{-1}\Delta A^{bid, }$.

In frictionless markets, the set of all self-financing trading strategies is a linear space because securities' prices are not influenced by the direction of trading.
This is no longer the case if the direction of trading matters: the strategy $\phi+\psi$ may not be self-financing even if $\phi$ and $\psi$ are self-financing.
Intuitively this is true because transaction costs can be avoided whenever $\phi^j_t \psi^j_t<0$ by combining orders.
However, the strategy $(\theta^0, \phi^1+\psi^1, \phi^2+\psi^2, \dots, \phi^N+\psi^N)$ can enjoy the self-financing property if the units in the savings account $\theta^0$ are properly adjusted.
The next lemma shows that such $\theta^0$ exists, is unique, and satisfies $\phi^0+\psi^0\leq \theta^0$.

\begin{proposition}\label{Prop:TradingStrategiesNonLinear}
  Let $\psi$ and $\phi$ be any two self-financing trading strategies with $V_0(\psi)=V_0(\phi)=0$.
   Then there exists a unique predictable process $\theta^0$ such that the trading strategy $\theta$ defined as $\theta:=(\theta^0, \phi^1+\psi^1, \dots, \phi^N+\psi^N)$ is self-financing with $V_0(\theta)=0$.
   Moreover, $\phi^0+\psi^0 \leq \theta^0$.
\end{proposition}

\begin{proof}
  The trading strategies $\phi$ and $\psi$ are self-financing, so by definition we have that
  \begin{align}\label{Eq:Phi}
 B_{t-1}\Delta  \phi^{0}_{t}+&\sum_{j=1}^N\Delta  \phi^{ j}_{t}\Big(1_{\{\Delta  \phi^j_t \geq 0\}}P^{ask, j}_{t-1}+1_{\{\Delta  \phi^j_t < 0\}}P^{bid, j}_{t-1} \Big)\\ \notag
 & \quad =  \sum_{j=1}^N \phi^{ j}_{t-1}\Big(1_{\{ \phi^j_{t-1} \geq 0\}} \Delta  A^{ask, j}_{t-1}+1_{\{ \phi^j_{t-1} < 0\}} \Delta  A^{bid, j}_{t-1}\Big)
  \end{align}
 and
    \begin{align}\label{Eq:Psi}
 B_{t-1}\Delta  \psi^{0}_{t}+&\sum_{j=1}^N\Delta  \psi^{ j}_{t}\Big(1_{\{\Delta  \psi^j_t \geq 0\}}P^{ask, j}_{t-1}+1_{\{\Delta  \psi^j_t < 0\}}P^{bid, j}_{t-1}\Big)\\ \notag
 & \quad =  \sum_{j=1}^N \psi^{ j}_{t-1}\Big(1_{\{ \psi^j_{t-1} \geq 0\}} \Delta A^{ask, j}_{t-1}+1_{\{ \psi^j_{t-1} < 0\}} \Delta A^{bid, j}_{t-1}\Big)
  \end{align}
 for $t=2, 3, \dots, T$.
By adding equations \eqref{Eq:Phi} and \eqref{Eq:Psi}, and rearranging terms we see that
\begin{align}\label{Eq:SumPhiPsi}
\psi^{0}_{t}+\phi^{0}_{t}&=\psi^0_{t-1}+\phi^0_{t-1}+B^{-1}_{t-1}\Bigg(-\sum_{j=1}^NP^{ask, j}_{t-1}\big( \Delta  \phi^{ j}_{t}+\Delta  \psi^{ j}_{t}\big)\\ \notag
 & \qquad  + \sum_{j=1}^N\big(\phi^{ j}_{t-1}+ \psi^{ j}_{t-1}\big)\Delta  A^{ask, j}_{t-1}\\ \notag
 &\qquad   +\sum_{j=1}^N\big(P^{ask, j}_{t-1}-P^{bid, j}_{t-1}\big)\Big(1_{\{\Delta  \phi^j_t < 0\}} \Delta  \phi^{ j}_{t}+1_{\{\Delta  \psi^j_t < 0\}} \Delta  \psi^{ j}_{t}\Big)\\ \notag
& \qquad   +\sum_{j=1}^N\Big(1_{\{ \phi^j_{t-1} < 0\}}\phi^{j}_{t-1}+ 1_{\{ \psi^j_{t-1} < 0\}}\psi^{j}_{t-1}\Big) \big(\Delta A^{bid, j}_{t-1}-\Delta  A^{ask, j}_{t-1}\big)\Bigg) \notag
  \end{align}
 for $t=2, 3, \dots, T$.
Now recursively define the process $\theta^0$ as
\begin{align*}
  \theta^0_1:&=-\sum_{j=1}^N\Big (1_{\{\phi^j_1\geq0 \}} \phi^j_1+1_{\{\psi^j_1\geq0 \}} \psi^j_1\Big)P^{ask, j}_0
  -\sum_{j=1}^N \Big(1_{\{\phi^j_1 <0 \}}\phi^j_1+1_{\{\psi^j_1 <0 \}}\psi^j_1\Big)P^{bid, j}_0,
\end{align*}
and
\begin{align}\label{Eq:Theta}
 \theta^{0}_{t}&:=\theta^0_{t-1}+B^{-1}_{t-1}\Bigg(-\sum_{j=1}^N \Delta \big(\phi^{ j}_{t}+\psi^{ j}_{t}\big)P^{ask, j}_{t-1}+\sum_{j=1}^N\big(\phi^{ j}_{t-1}+\psi^{ j}_{t-1}\big)\Delta  A^{ask, j}_{t-1}\\\notag
&\qquad +\sum_{j=1}^N1_{\{\Delta  (\phi^{ j}_{t}+\psi^{ j}_{t}) < 0\}} \Delta \big(\phi^{ j}_{t}+\psi^{ j}_{t}\big)\big(P^{ask, j}_{t-1}-P^{bid, j}_{t-1}\big)\\\notag
 & \qquad +\sum_{j=1}^N 1_{\{ \phi^{ j}_{t-1}+\psi^{ j}_{t-1}< 0\}}\big(\phi^{ j}_{t-1}+\psi^{ j}_{t-1}\big)\big(\Delta A^{bid, j}_{t-1}-\Delta  A^{ask, j}_{t-1}\big)\Bigg)\notag
\end{align}
 for $t=2, 3, \dots, T$.
It follows that $\theta^0$ is unique and is self-financing.
By definition, the trading strategy $\theta:=(\theta^0, \phi^1+\psi^1, \dots, \phi^N+\psi^N)$ is self-financing.
Subtracting \eqref{Eq:SumPhiPsi} from \eqref{Eq:Theta} yields
\begin{align}\label{Eq:SubtractThetaFromPhiPsi}
 \theta^{0}_{t}-(\phi^0_t+\psi^0_t)&=\theta^0_{t-1}-\big(\phi^0_{t-1}+\psi^0_{t-1}\big)\\ \notag
& \;  +\sum_{j=1}^N1_{\{\Delta\phi^{ j}_{t}+\psi^{ j}_{t}) < 0\}} \Delta\big(\phi^{ j}_{t}+\psi^{ j}_{t}\big)\big(P^{ask, j}_{t-1}-P^{bid, j}_{t-1}\big)\\ \notag
& \;  -\sum_{j=1}^N\Big(1_{\{\Delta  \phi^j_t < 0\}}\Delta  \phi^{ j}_{t}+1_{\{\Delta  \psi^j_t < 0\}}\Delta  \psi^{ j}_{t}\Big)\big(P^{ask, j}_{t-1}-P^{bid, j}_{t-1}\big)\\ \notag
&  \; +\sum_{j=1}^N \big(\phi^{ j}_{t-1}+\psi^{ j}_{t-1}\big)1_{\{ \phi^{ j}_{t-1}+\psi^{ j}_{t-1}< 0\}}\big(\Delta A^{bid, j}_{t-1}-\Delta  A^{ask, j}_{t-1}\big)\\ \notag
 & \;  -\sum_{j=1}^N\big(\phi^{j}_{t-1}1_{\{ \phi^j_{t-1} < 0\}}+ \psi^{j}_{t-1}1_{\{ \psi^j_{t-1} < 0\}}\big)\big(\Delta A^{bid, j}_{t-1}-\Delta  A^{ask, j}_{t-1}\big)\notag
  \end{align}
for $t=2, 3, \dots, T$.
It is straightforward to verify that  the inequality
\begin{equation}\label{eq:ConvexInequality}
1_{\{X <0\}}X+1_{\{Y<0\}}Y \leq 1_{\{X+Y <0\}}(X+Y)
\end{equation}
 holds for any random variables $X$ and $Y$.
 Moreover, the inequalities $P^{bid} \leq P^{ask}$  and $\Delta A^{ask} \leq \Delta A^{bid}$ hold by assumption.
Hence, \eqref{Eq:SubtractThetaFromPhiPsi} reduces to
\begin{equation}\label{Eq:DifferenceEqTheta}
 \theta^{0}_{t}-\big(\phi^0_t+\psi^0_t\big)\geq \theta^0_{t-1}-\big(\phi^0_{t-1}+\psi^0_{t-1}\big)
 \end{equation}
 for $t=2, 3, \dots, T$.
  Since $V_0(\phi)=V_0(\psi)=0$,
it follows that $\theta^0_1 = \phi^0_1 +\psi^0_1$ and $V_0(\theta)=V_0(\phi)+V_0(\psi)=0$.
After recursively solving \eqref{Eq:DifferenceEqTheta}, we conclude that $\theta^0_t \geq\phi^0_t+\psi^0_t$ for all $t \in \cT^*$.
\end{proof}

The next result is the natural extension of the previous proposition to value processes.
It is intuitively true since some transaction costs may be avoided by combining orders.

\begin{theorem}\label{Arbitrage/LemmaNonLinearSF}
Let $\phi$ and $\psi$ be any two self-financing trading strategies such that $V_0(\phi)=V_0(\psi)=0$.
 There exists a unique predictable process $\theta^0$ such that the trading strategy defined as $\theta:=(\theta^0, \phi^1+\psi^1, \dots, \phi^N+\psi^N)$ is self-financing with $V_0(\theta)=0$, and $V_T(\theta)$ satisfies
\begin{equation*}
  V_T(\phi)+V_T(\psi) \leq V_T(\theta).
\end{equation*}
\end{theorem}

\begin{proof}
Let $\phi$ and $\psi$ be self-financing trading strategies.
By applying Proposition~\ref{Intro/LemmaSelfFinancing} and rearranging terms, we may write
\begin{align}
 V_T(\phi)+V_T(\psi)&=\sum_{u=1}^T(\phi^0_u+\psi^0_u)\Delta B_u \notag\\
& +\sum_{j=1}^N\big(\phi^{j}_T+\psi^{j}_T\big)\Big( 1_{\{\phi^j_T +\psi^j_T \geq 0\}}P^{bid, j}_T + 1_{\{\phi^j_T +\psi^j_T < 0\}} P^{ask, j}_T\Big)\notag\\
& -\sum_{j=1}^N\sum_{u=1}^T\big(\Delta\phi^j_u + \Delta \psi^j_u \big)\Big(1_{\{\Delta\phi^j_u +\Delta \psi^j_u \geq 0\}}P^{ask, j}_{u-1}
 +1_{\{\Delta\phi^j_u +\Delta \psi^j_u  < 0\}}P^{bid, j}_{u-1}\Big)\notag\\
& +\sum_{j=1}^N \sum_{u=1}^T\big(\phi^{ j}_u+\psi^{ j}_u\big)\Big(1_{\{\phi^j_u +\psi^j_u \geq 0\}}\Delta A^{ask,j}_u
 +1_{\{\phi^j_u + \psi^j_u < 0\}}\Delta A^{bid,j}_u\Big) -C^1- C^2,\label{eq:SumValues}
\end{align}
where $C^1$ is  defined as
\begin{align*}
&  C^1:=\sum_{j=1}^N \Bigg(\Big( 1_{\{\phi^j_T \geq 0\}}\phi^{ j}_T+1_{\{\psi^j_T \geq 0\}}\psi^{ j}_T\Big)P^{ask, j}_T
  +\Big(1_{\{\phi^j_T < 0\}} \phi^{ j}_T+1_{\{\psi^j_T <0\}}\psi^{ j}_T\Big)P^{bid, j}_T\Bigg)\\
&   -\sum_{j=1}^N\sum_{u=1}^T\Big(1_{\{\Delta\phi^j_u < 0\}}\Delta\phi^j_u+ 1_{\{\Delta\psi^j_u < 0\}}\Delta\psi^j_u\Big)P^{ask, j}_{u-1}\\
 &-\sum_{j=1}^N\sum_{u=1}^T\Big( 1_{\{\Delta\phi^j_u \geq 0\}}\Delta\phi^j_u+ 1_{\{\Delta\psi^j_u \geq 0\}}\Delta\psi^j_u\Big)P^{bid, j}_{u-1}\\
&+\sum_{j=1}^N\sum_{u=1}^T\Bigg(\Big(1_{\{\phi^j_u \geq 0\}}\phi^{ j}_u+1_{\{\psi^j_u \geq 0\}}\psi^{ j}_u\Big)\Delta A^{bid,j}_u
+\Big(1_{\{\phi^j_u < 0\}}\phi^{ j}_u+1_{\{\psi^j_u < 0\}}\psi^{ j}_u\Big)\Delta A^{ask, j}_u\Bigg),
\end{align*}
and $C^2$  is defined as
\begin{align*}
  C^2&:= -\sum_{j=1}^N \big(\phi^{ j}_T + \psi^{ j}_T\big)\Big(1_{\{\phi^j_T +\psi^j_T < 0\}}P^{bid, j}_T
  +1_{\{\phi^j_T +\psi^j_T \geq 0\}} P^{ask, j}_T\Big) \\
& +\sum_{j=1}^N\sum_{u=1}^T\big(\Delta\phi^j_u +\Delta \psi^j_u \big)\Big(1_{\{\Delta\phi^j_u +\Delta \psi^j_u \geq 0\}}P^{bid, j}_{u-1}
+1_{\{\Delta\phi^j_u +\Delta \psi^j_u < 0\}}P^{ask, j}_{u-1}\Big)\\
 &-\sum_{j=1}^N\sum_{u=1}^T\big(\phi^{ j}_u+\psi^{ j}_u\big)\Big(1_{\{\phi^j_u + \psi^j_u \geq 0\}}\Delta A^{bid,j}_u
 +1_{\{\phi^j_u + \psi^j_u < 0\}}\Delta A^{ask, j}_u\Big).
\end{align*}
By Proposition~\ref{Prop:TradingStrategiesNonLinear}, there exists a unique predictable process $\theta^0$ such that the trading strategy defined as $\theta:=(\theta^0, \phi^1+\psi^1, \dots, \phi^N+\psi^N)$ is self-financing with $V_0(\theta)=0$ and satisfies $\phi^0+\psi^0 \leq \theta^0$.
In view of Proposition~\ref{Intro/LemmaDiscountedSelfFinancing},
since $\theta$ is self-financing, it follows that
\begin{align}\label{eq:ValueProcessTheta}
  V_T(\theta)& =\sum_{u=1}^T\theta^0_u\Delta B_u +\sum_{j=1}^N\big(\phi^{j}_T+\psi^{j}_T\big)\Big( 1_{\{\phi^j_T +\psi^j_T \geq 0\}}P^{bid, j}_T \notag
 +1_{\{\phi^j_T +\psi^j_T < 0\}} P^{ask, j}_T\Big)\\ \notag
&   -\sum_{j=1}^N\sum_{u=1}^T\big(\Delta\phi^j_u +\Delta \psi^j_u  \big)\Big(1_{\{\Delta\phi^j_u +\Delta \psi^j_u  \geq 0\}}P^{ask, j}_{u-1}
 +1_{\{\Delta\phi^j_u +\Delta \psi^j_u  < 0\}}P^{bid, j}_{u-1}\Big)\\
& +\sum_{j=1}^N \sum_{u=1}^T\big(\phi^{ j}_u+\psi^{ j}_u\big)\Big(1_{\{\phi^j_u +\psi^j_u \geq 0\}}\Delta A^{ask,j}_u
 +\sum_{u=1}^T1_{\{\phi^j_u + \psi^j_u < 0\}}\Delta A^{bid,j}_u\Big).
\end{align}
Comparing equations \eqref{eq:SumValues} and \eqref{eq:ValueProcessTheta} we see that
\begin{equation}\label{eq:nonlinear1}
 V_T(\phi)+V_T(\psi)=V_T(\theta)+\sum_{u=1}^T(\phi^0_u+\psi^0_u-\theta^0_u)\Delta B_u-C^1- C^2.
\end{equation}
According to \eqref{eq:ConvexInequality}, the random variable $C^1+C^2$ is nonnegative.
Moreover, since $\phi^0+\psi^0 \leq \theta^0$ and $\Delta B \geq 0$, it follows that $ \sum_{u=1}^T(\phi^0_u+\psi^0_u-\theta^0_u)\Delta B_u \leq 0$
From \eqref{eq:nonlinear1}, we conclude that $V_T(\phi)+V_T(\psi)\leq V_T(\theta)$.
\end{proof}

The following technical lemma, which easy to verify, will be used in the next section.

\begin{lemma} \label{Arbitrage/ValueConvergenceSF}
The following hold:
\begin{itemize}
\item
Let $Y^a$ and $Y^b$ be any random variables, and suppose $X^m$ is a sequence of $\bR$-valued random variables converging a.s. to $X$.
 Then $1_{\{X^m \geq 0\}}X^mY^b+1_{\{X^m <0\}}X^mY^a$ converges a.s. to $1_{\{X \geq 0\}}XY^b+1_{\{X <0\}}XY^a.$
\item
  If a sequence of trading strategies $\phi^m$ converges a.s. to $\phi$, then $V(\phi^m)$ converges a.s. to $V(\phi)$.
  \end{itemize}
\end{lemma}

\subsection{The set of claims that can be superhedged at zero cost}
For all $t \in \cT$, denote by $L^0(\Omega, \cF_{t}, \bP\,; \bR^{(N+1)})$ the space of all ($\bP$-equivalence classes of) $\bR^{(N+1)}$-valued, $\cF_{t}$-measurable random variables.
We equip $L^0(\Omega, \cF_t, \bP\, ; \bR)$ with the topology of convergence in measure $\bP$.
Also, let $\cS$ be the set of all self-financing trading strategies.
For the sake of conciseness, we will refer to sets that are closed with respect to convergence in measure $\bP$ simply as $\bP$-closed.

We define the sets
\begin{align*}
&\cK:=\big\{ V^*_T(\phi):  \, \phi \in \cS, \; V_0(\phi)=0\big\},\\
&L^0_+(\Omega, \cF_T, \bP\,; \bR):=\big\{X \in L^0(\Omega, \cF_T, \bP\,; \bR):\; X \geq 0\big\},\\
&\cK-L^0_+(\Omega, \cF_T, \bP\,; \bR):=\big\{Y-X: \; Y \in \cK \;  \text{and} \; X \in L^0_+(\Omega, \cF_T, \bP\,; \bR)\big\}.
\end{align*}
The set $\cK$ is the \emph{set of attainable claims at zero cost}.
On the other hand, $\cK-L^0_+(\Omega, \cF_T, \bP\,; \bR)$ is \emph{the set of claims that can be superhedged at zero cost}: for any $X \in \cK-L^0_+(\Omega, \cF_T, \bP\,; \bR)$, there exists an attainable value at zero cost $K \in \cK$ so that $X \leq K$.

The following lemma asserts that the set of claims that can be superhedged at zero cost is a convex cone.

\begin{lemma}\label{Arbitrage/LemmaConvexCone}
The set $\cK-L^0_+(\Omega, \cF_T, \bP\,; \bR)$ is a convex cone.
\end{lemma}
\begin{proof}
Let $Y^1, Y^2 \in \cK-L^0_+(\Omega, \cF_T, \bP\,; \bR)$.
Then there exist $K^1, K^2 \in \cK$ and $Z^1, Z^2 \in L^0_+(\Omega, \cF_T, \bP\,; \bR)$ such that $Y^1=K^1-Z^1$ and $Y^2=K^2-Z^2$.
By definition of $\cK$, there exists $\phi, \psi \in \cS$ with $V_0(\phi)=V_0(\psi)=0$ such that $K^1=V^*_T(\phi)$ and $K^2=V^*_T(\psi)$.
We will prove that for any positive scalars $\alpha_1$ and $\alpha_2$ the following holds
\begin{equation*}
\alpha_1 (V^*_T(\phi)-Z^1)+\alpha_2 (V^*_T(\psi)-Z^2) \in \cK-L^0_+(\Omega, \cF_T, \bP\,; \bR),
\end{equation*}
or, equivalently, that there exists $K \in \cK$ such that
\begin{equation*}
\alpha_1 V^*_T(\phi)+\alpha_2 V^*_T(\psi)-\alpha_1 Z^1-\alpha_2 Z^2 \leq K.
\end{equation*}
 The value process is positive homogeneous, so $\alpha_1V^*_T(\phi)+\alpha_2V^*_T(\psi)=V^*_T(\alpha_1 \phi)+ V^*_T(\alpha_2\psi)$.
According to Theorem~\ref{Arbitrage/LemmaNonLinearSF}, there exists a unique predictable process $\theta^0$ such that the trading strategy defined as $\theta:=(\theta^0, \alpha_1\phi^{1}+\alpha_2\psi^{1}, \dots, \alpha_1\phi^{N}+\alpha_2\psi^{N})$ is self-financing with $V_0(\theta)=0$, and satisfies $V^*_T(\alpha_1\phi)+V^*_T(\alpha_2\psi) \leq V^*_T(\theta)$.
By definition, we have $V^*_T(\theta) \in \cK$.
Since
\begin{align*}
\alpha_1 V^*_T( \phi)+\alpha_2 V^*_T(\psi)-\alpha_1 Z^1-\alpha_2 Z^2  &=V^*_T(\alpha_1 \phi)+V^*_T(\alpha_2 \psi)-\alpha_1 Z^1-\alpha_2 Z^2 \\
& \leq V^*_T(\theta)-\alpha_1 Z^1-\alpha_2 Z^2 \leq V^*_T(\theta),
\end{align*}
we conclude that the  claim holds.
\end{proof}

\begin{remark}
  The set $\cK$ is not necessarily a convex cone.
 To see this, lets suppose that $T=1$, $\cJ=\{0, 1\}$, and $r=0$.
Consider the trading strategies $\phi=\{\phi^0, 1\}$ and $\psi=\{\psi^0, -1\}$, where  $\phi^0$ and $\psi^0$  are chosen so that $V_0(\phi)=V_0(\psi)=0$.
By definition,  $V_1(\phi), V_1(\psi) \in \cK$.
However,
  \begin{equation*}
    V^*_1(\phi)+V^*_1(\psi)=P^{bid}_1-P^{ask}_1+A^{ask}_1-A^{bid}_1+P^{bid}_0-P^{ask}_0,
  \end{equation*}
  is generally not in the set $\cK$.
\end{remark}

\section{The no-arbitrage condition}\label{section:Arbitrage}

We begin by introducing the definition of the no-arbitrage condition.

\begin{definition}\label{def:arbitrage}
\emph{The no-arbitrage condition} \textbf{(NA)} is satisfied if for each $\phi \in \cS$ such that $V_0(\phi)=0$ and $V_T(\phi) \geq0$, we have $V_T(\phi)=0$.
\end{definition}
In the present context, \textbf{NA} has the usual interpretation that ``it is impossible to make something out of nothing.''
The next lemma provides us equivalent conditions  to \textbf{NA} in terms of the set of attainable claims at zero cost, and also in terms of the set of claims that can be superhedged at zero cost.
They are straightforward to verify.

\begin{lemma}\label{Arbitrage/LemmaEquArb}
The following conditions are equivalent:
\begin{itemize}
\item[(i)]\emph{\textbf{NA}} is satisfied.
\item[(ii)] $\big(\cK-L^0_+(\Omega, \cF_T, \bP\,; \bR)\big) \cap L^0_+(\Omega, \cF_T, \bP\,; \bR)=\{0\}$.
\item [(iii)] $\cK \cap L^0_+(\Omega, \cF_T, \bP\,; \bR)=\{0\}.$
\end{itemize}
\end{lemma}

We proceed by defining The Efficient Friction Assumption in our context (cf. Kabanov et al.~\cite{KabanovRasonyiStricker2002}).

\smallskip
\noindent
\textbf{The Efficient Friction Assumption (EF):}
  \begin{equation}\label{eq:AssumptionRedundant}
    \big\{\phi \in \cS: \; V_0(\phi)=V_T(\phi)=0 \big\}= \{0\}.
  \end{equation}

Note that if \eqref{eq:AssumptionRedundant} is satisfied, then for each $\phi \in \cS$, we have $V_0(\phi)=V_T(\phi)=0$ if and only if $\phi=0$.
The efficient friction assumption, which was introduced by Kabanov et al.~[KRS02], states that the only portfolio that can be liquidated into the zero portfolio that is available at zero price is the zero portfolio.
In the present context, \textbf{EF} has the same interpretation: the only zero-cost, self-financing strategy that can be liquidated into the zero portfolio is the zero portfolio.

We will denote by \textbf{NAEF} \emph{the no-arbitrage condition under the efficient friction assumption}.

 In what follows, we denote by $\cP$  the set of all $\bR^{N}$-valued, $\bF$-predictable processes.
 Also, we define the mapping
  \begin{align}
  F(\phi)&:=\sum_{j=1}^N  \phi^{j}_T\big(1_{\{ \phi^{j}_T \geq 0\}}P^{bid, j, *}_T+ 1_{\{ \phi^{j}_T < 0\}}P^{ask, j, *}_T\big)\notag\\
& \  -\sum_{j=1}^N\sum_{u=1}^T \Delta \phi^{ j}_u\big(1_{\{\Delta \phi^{j}_u \geq 0\}}P^{ask, j, *}_{u-1}+1_{\{\Delta \phi^{ j}_u< 0\}} P^{bid, j, *}_{u-1}\big)\notag\\
& \ +\sum_{j=1}^N \sum_{u=1}^T \phi^{ j}_u \big(1_{\{ \phi^{j}_u \geq 0\}} A^{ask,j, *}_u+1_{\{ \phi^{ j}_u < 0\}}A^{bid,j, *}_u\big)\label{eq:DefMappingF}
\end{align}
for all $\bR^N$-valued stochastic processes
\begin{equation*}
 (\phi_s)_{s=1}^T \in L^0(\Omega, \cF_T, \bP; \bR^N) \times \cdots \times L^0(\Omega, \cF_T, \bP; \bR^N) ,
 \end{equation*}
 and let $\bK:=\{F(\phi): \phi \in \cP\}$.
In view of Proposition~\ref{Intro/LemmaDiscountedSelfFinancing}, we note that $V^*_T(\phi)=V_0(\phi)+F(\phi)$ for all self-financing trading strategies $\phi$

\begin{remark}\mbox{}
\begin{itemize}
\item[(i)]
Note that $F$ is defined on the set of all $\bR^N$-valued stochastic processes.
On the contrary, the value process is defined on the set of trading strategies, which are $\bR^{N+1}$-valued predictable processes.
\item[(ii)]
The set $\bK$ has the same financial interpretation as the set $\cK$.
We introduce the set $\bK$ because it is more convenient to work with from the mathematical point of view.
\item[(iii)]
$F(\alpha \phi)=\alpha F(\phi)$ for any nonnegative random variable $\alpha$.
\end{itemize}
\end{remark}

The next results provides an equivalent condition for \textbf{EF} to hold.

\begin{lemma}\label{lemma: EFequivalent}
The efficient friction assumption \textbf{\emph{(EF)}} is satisfied if and only if $\{\psi \in \cP: F(\psi)=0\}=\{0\}$.
\end{lemma}

\begin{proof}
 Let $\psi \in \cP$ be such that $F(\psi)=0$.
We define the trading strategy\\ $\phi:=(\phi^0, \psi^1, \dots, \psi^N)$, where $\phi^0$ is chosen so that $\psi \in \cS$ and $V_0(\phi)=0$.
we see that $V^*_T(\phi)=F(\psi)$, which gives us $V^*_T(\phi)=0$.
\textbf{EF} is satisfied, so $\phi^j=0$ for $j=0, \dots, N$, which in particular implies that $\psi^j=0$ for $j=1,\dots, N$.

Conversely, suppose \textbf{EF} holds, and fix $\phi \in \cS$ so that $V_0(\phi)=V^*_T(\phi)=0$.
Define the predictable process $\psi^j:=\phi^j$ for $j=1 , \dots, N$.
By Proposition~\ref{Intro/LemmaDiscountedSelfFinancing} and the definition of $F$, it is true that $F(\psi)=V^*_T(\phi)$.
Thus, $F(\psi)=0$.
By assumption, we have that  $\psi^j=0$ for $j=1, \dots, N$, which implies $\phi^j=0$ for $j=1, \dots, N$.
From the definition of $V_0(\phi)$ and because $V_0(\phi)=0$,
it follows that $\phi^0_1=0$.
Since $\phi \in \cS$,
we may recursively solve for $\phi^0_2, \dots, \phi^0_T$ to deduce that $\phi^0_t=0$ for $t=2, \dots, T$.
Hence, $\phi^j=0$ for $j=0,1, \dots, N$.
\end{proof}

\begin{lemma}\label{lemma: Kequivalent}
We have that $\cK=\bK$.
\end{lemma}

\begin{proof}
The claim follows from  Proposition~\ref{Intro/LemmaDiscountedSelfFinancing}.
\end{proof}

\subsection{Closedness property of the set of claims that can be superhedged at zero cost}\label{subsection:Closedness}
In this section, we prove that the set of claims that can be superhedged at zero cost, $\cK-L^0_+(\Omega, \cF_T, \bP\,; \bR)$,  is $\bP$-closed whenever \textbf{NAEF} is satisfied.
This property plays a central role in the proof of the First Fundamental Theorem of Asset Pricing (Theorem~\ref{Theorem: FFTAP}).

We will denote by $\| \cdot\|$ the Euclidean norm on $\bR^N$.

Let us first recall the following lemma from Schachermayer~\cite{Schachermayer2004}, which is closely related to Lemma 2 in Kabanov and Stricker~\cite{KabanovStricker2001}.

\begin{lemma}\label{Arbitrage/MeasurableSelectionShach}
For a sequence of random variables $X^m \in L^0(\Omega, \cF, \bP; \bR^N)$ there is a strictly increasing sequence of positive, integer-valued, $\cF$-measurable random  variables $\tau^m$ such that $X^{\tau^m}$ converges a.s. in the one-point-compactification $\bR^N \cup \{\infty\}$ to some random variable $X \in L^0(\Omega, \cF, \bP; \bR^N\cup \{\infty\})$.
Moreover, we may find the subsequence  such that $\|X\|=\limsup_m\|X^{m}\|$, where $\|\infty\|=\infty$.
\end{lemma}

The next result extends the previous lemma to processes.

\begin{lemma}\label{Arbitrage/MeasurableSelectionShachProcess}
Let $\cF^i$ be a $\sigma$-algebra, and  $Y^m_i \in L^0(\Omega, \cF^i, \bP; \bR^N)$ for  $i=1, \dots, M$.
 Suppose that $\cF^i \subseteq \cF^j$ for all $i \leq j$, and that $Y^m_i$ satisfies $\limsup_m \|Y^m_i\|<\infty$ for $i=1, \dots, M$.
 Then there is a strictly increasing sequence of positive, integer-valued, $\cF^{M}$-measurable random  variables $\tau^m$ such that, for $i=1, \dots, M$, the sequence $Y^{\tau^m}_i$ converges a.s. to some $Y_i \in L^0(\Omega, \cF^i, \bP; \bR^N)$.
\end{lemma}

\begin{proof}
  We first apply Lemma~\ref{Arbitrage/MeasurableSelectionShach} to the random variable $Y^m_1$: there exists a strictly increasing sequence of positive, integer-valued, $\cF^{1}$-measurable random variables $\tau^m_1$ such that $\{\tau^1_1(\omega), \tau^2_1(\omega), \dots, \} \subseteq\bN$ for $\omega \in \Omega$, and  $Y^{\tau^m_1}_1$ converges a.s. to some $Y_1 \in L^0(\Omega, \cF^1, \bP; \bR^N)$.
Since $\limsup_m \|Y^m_2\| <\infty$, we also have that $\limsup_m \|Y^{\tau^m_1}_2\| <\infty$.
Moreover,  $Y^{\tau^m_1}_2 \in L^0(\Omega, \cF^2, \bP; \bR^N)$ since $\cF^1 \subseteq \cF^2$.
Therefore, we may apply Lemma~\ref{Arbitrage/MeasurableSelectionShach} to the sequence $Y^{\tau^m_1}_2$ to find a strictly increasing sequence of positive, integer-valued, $\cF^{2}$-measurable random variables $\tau^m_2$ such that
\begin{equation}\label{eq:Meas0}
\{\tau^1_2(\omega), \tau^2_2(\omega), \dots \} \subseteq \{\tau^1_1(\omega), \tau^2_1(\omega), \dots \} \subseteq \bN, \quad \textrm{a.e. }\omega \in \Omega,
 \end{equation}
 and   $Y^{\tau^m_2}_2$ converges a.s. to some $Y_2 \in L^0(\Omega, \cF^2, \bP; \bR^N)$.
 From \eqref{eq:Meas0}, the sequence $Y^{\tau^m_2}_1$ converges a.s. to $Y_1$.

 We may continue by recursively repeating the argument above to the sequences $Y^m_i$, for $i =3, \dots, M$, to find  strictly increasing sequences of positive, integer-valued, $\cF^{i}$-measurable random variables $\tau^m_i$ such that
 \begin{equation}\label{eq:Meas1}
\{\tau^1_i(\omega), \tau^2_i(\omega), \dots \}\subseteq \cdots \subseteq \{\tau^1_1(\omega), \tau^2_1(\omega), \dots \} \subseteq \bN, \quad  \textrm{a.e. }\omega \in \Omega,
 \end{equation}
 and   $Y^{\tau^m_i}_i$ converges a.s. to some $Y_i \in L^0(\Omega, \cF^{i}, \bP; \bR^N)$.
 Because of \eqref{eq:Meas1}, we see that $Y^{\tau^m_M}_i$ converges a.s. to $Y_i$ for $i=1, \dots, M$.
 Therefore, $\tau^m:=\tau^m_M$ defines the desired sequence.
\end{proof}

We proceed by proving a technical lemma.

\begin{lemma}\label{lemma:NormalizedSequence}
Let $\cF^i$ be a $\sigma$-algebra, and  $Y^m_i \in L^0(\Omega, \cF^i, \bP; \bR^N)$ for  $i=1, \dots, M$.
 Suppose that $\cF^i \subseteq \cF^j$ for all $i \leq j$, and that  there exists $k \in \{1, \dots, M\}$ and  $\Omega' \subseteq \Omega$ with $\bP(\Omega')>0$ such that $\limsup_m \|Y^m_k(\omega)\|=\infty$ for a.e. $\omega \in \Omega'$, and $\limsup_m \|Y^m_i(\omega)\|<\infty$ for $i=1, \dots, k-1$ and for a.e. $\omega \in \Omega$.
Then there exists  a strictly increasing sequence of positive, integer-valued, $\cF^{k}$-measurable random  variables $\tau^m$ such that
$\lim_m\|Y^{\tau^m}_k(\omega)\|=\infty$, for a.e. $\omega \in \Omega'$, and
\footnote{We take $X^{m}_i(\omega)=0$ whenever $\|Y^{\tau^m(\omega)}_{k}(\omega)\|=0$. We will take the convention $x/0=0$ throughout this section.}
 \begin{equation*}
 X^{m}_i(\omega):=1_{\Omega'}(\omega)\frac{Y^{\tau^m(\omega)}_i(\omega)}{\|Y^{\tau^m(\omega)}_k(\omega)\|}, \quad \omega \in \Omega,\; i=1, \dots, M,
\end{equation*}
satisfies $\lim_m X^m_i(\omega)=0$, for  $i=1, \dots, k-1$ and for a.e. $\omega \in \Omega$
\end{lemma}

\begin{proof}
Since $\limsup_m  \|Y^m_k(\omega)\|= \infty$ for a.e. $\omega \in \Omega'$, we may apply Lemma~\ref{Arbitrage/MeasurableSelectionShach} to the sequence $Y^m_k$ to find a strictly increasing sequence of positive, integer-valued, $\cF^k$-measurable random variables $\tau^m$ so that $\|Y^{\tau^m(\omega)}_k(\omega)\|$ diverges for a.e. $\omega \in \Omega'$.

Because $\limsup_m \|Y^m_i\|<\infty$ for $i=1, \dots, k-1$, we have  $\limsup_m \|Y^{\tau^m}_i\|<\infty$ for $i=1, \dots, k-1$.
Now since  $\|Y^{\tau^m(\omega)}_{k}(\omega)\|$ diverges for a.e. $\omega \in \Omega'$,
\begin{equation*}
\lim_{m \to \infty} \|X^{m}_i(\omega)\|=1_{\Omega'}(\omega) \frac{\|Y^{\tau^m(\omega)}_i(\omega)\|}{\|Y^{\tau^m(\omega)}_k(\omega)\|}=0, \quad \textrm{a.e. } \omega \in \Omega,\; i =1, \dots, k-1.
\end{equation*}
Thus, $\|X^{m}_i\|$ converges a.s. to 0 for $i=1, \dots, k-1$, which implies that $X^{m}_i$ converges a.s. to 0  for $i=1, \dots, k-1$.
Hence, the claim holds.
\end{proof}

We are now ready to prove the crucial result in this paper.

\begin{theorem}\label{lemma:KClosed}
 If the no-arbitrage condition under the efficient friction assumption\\ \emph{\textbf{(NAEF)}} is satisfied, then the set $\cK-L^0_+(\Omega, \cF_T, \bP; \bR)$ is $\bP$-closed.
\end{theorem}

\begin{proof}
According to Lemma~\ref{lemma: Kequivalent}, we may equivalently prove that $\bK-L^0_+(\Omega, \cF_T, \bP; \bR)$ is $\bP$-closed.
Suppose that $X^m \in \bK-L^0_+(\Omega, \cF_T, \bP; \bR)$ converges in probability to $X$.
  Then there exists a subsequence $X^{k_m}$ of $X^{k}$ so that $X^{k_m}$ converges a.s. to $X$.
 With a slight abuse of notation, we will denote by $X^{m}$ the sequence $X^{k_m}$ in what follows.
 By the definition of $\bK-L^0_+(\Omega, \cF_T, \bP; \bR)$, there exists $Z^m \in L^0_+(\Omega, \cF_T, \bP; \bR)$ and $\phi^m \in \cP$ so that
 \begin{equation}\label{eq:KClose11}
 X^m=F(\phi^m)-Z^m.
 \end{equation}

We proceed the proof in two steps.
In the first step, we show  by contradiction that $\limsup_m\| \phi^m_{s} \| <\infty$ for all $s \in \cT^*$.

\noindent
\textbf{Step 1a:}
Let us assume that $\limsup_m\| \phi^m_{s} \| <\infty$ for all $s \in \cT^*$ does not hold.
Then
\begin{equation*}
\cI^0:=\Big\{ s\in \cT^*: \exists \; \Omega' \subseteq \Omega \textrm{ such that } \bP(\Omega')>0,\; \limsup_{m \to \infty}\| \phi^m_{s}(\omega) \|=\infty  \textrm{ for a.e. }  \omega \in \Omega' \Big\}
\end{equation*}
is nonempty.
Let  $t_0:=\min \cI^0$, and define the $\cF_{t_0-1}$-measurable set
\begin{equation*}
  E^{0}:=\big\{ \omega \in \Omega : \limsup_{m \to \infty} \| \phi^m_{t_0}(\omega)\|=\infty\}.
\end{equation*}
Note that $\bP(E^0)>0$ by assumption.
We now apply Lemma~\ref{lemma:NormalizedSequence} to $\phi^m$: there exists a strictly increasing sequence of positive, integer-valued, $\cF_{t_0-1}$-measurable random variables $\tau^m_0$ such that
\begin{equation}\label{eq:ClosePhiSupInfty}
\lim_{m \to \infty} \|\phi^{\tau^m_0(\omega)}_{t_0}(\omega)\|=\infty, \quad \textrm{a.e. }\omega \in E^0,
\end{equation}
and
\begin{equation}\label{eq:CloseNormalizedPhi}
  \psi^{m, (0)}_s:=1_{E^0}\frac{\phi^{\tau^m_0}_s}{\|\phi^{\tau^m_0}_{t_0}\|}, \quad s \in \cT^*,
\end{equation}
satisfies  $\lim_m \psi^{m, (0)}_s(\omega)=0$, for  $s=1, \dots, t_0-1$, for a.e. $\omega \in \Omega$.

We proceed as follows.

\smallskip
\noindent
\textbf{Recursively for} $i=1, \dots, T$

\noindent
\textbf{If }$\limsup_m \|\psi^{m, (i-1)}_s\|<\infty$ for all $s \in \{t_{i-1}+1, \dots, T\}$, then define $k:=i$ and $\varphi^m:=\psi^{m, (k-1)}$, and proceed to Step 1b.

\noindent
\textbf{Else}, define
\begin{align*}
 &t_i:=  \min \Big\{ s\in \{t_{i-1}+1, \dots, T\}: \exists \; \Omega' \subseteq E^{i-1} \textrm{ s.t. }\bP(\Omega')>0,\; \\
 & \qquad \qquad  \qquad \limsup_{m \to \infty} \|\psi^{m, (i-1)}_s(\omega)\| =\infty \textrm{ for a.e. } \omega \in \Omega'\Big\},
 \end{align*}
 and
 \begin{equation*}
E^i:=\big\{ \omega \in E^{i-1} : \limsup_{m \to \infty}\|\psi^{m, (i-1)}_{t_i}(\omega)\|=\infty\}.
\end{equation*}
  Next, apply Lemma~\ref{lemma:NormalizedSequence} to $\psi^{m, (i)}$:
  there exists a strictly increasing sequence of positive, integer-valued, $\cF_{t_i-1}$-measurable random variables $\tau^m_i$ such that
  \begin{equation}\label{eq:CloseNestedTau}
    \{\tau^1_i(\omega), \tau^2_i(\omega), \dots\}\subseteq \cdots  \subseteq  \{\tau^1_0(\omega), \tau^2_0(\omega), \dots\}, \quad \textrm{a.e. } \omega \in \Omega,
  \end{equation}
  the sequence $\psi^{\tau^m_i, (i-1)}_{t_i}$ satisfies
  \begin{equation}\label{eq:ClosePsiSupInfty}
  \lim_{m \to \infty} \|\psi^{\tau^m_i(\omega), (i-1)}_{t_i}(\omega)\|=\infty, \quad \textrm{a.e. }\omega \in E^i,
  \end{equation}
  and the sequence $\psi^{m, (i)}$ defined as
\begin{equation}\label{eq:CloseRecursion}
  \psi^{m, (i)}_s:=1_{E^i}\frac{ \psi^{\tau^m_i, (i-1)}_s}{\|\psi^{\tau^m_i, (i-1)}_{t_i}\|}, \quad  s \in \cT^*,
\end{equation}
satisfies  $\lim_m \psi^{m, (i)}_s(\omega)=0$ for  $s=1, \dots, t_i-1$, for a.e. $\omega \in \Omega$.

\noindent
\textbf{Repeat:} $i \rightarrow i+1$.
\smallskip

Given this construction, we define
\begin{align*}
   \beta^m_{i}(\omega)&:= \tau_i \circ \tau_{i+1} \circ \cdots \circ \tau^m_k(\omega),& \quad &i \in \{0, \dots, k\}, \; \omega \in \Omega, \\
  U^m(\omega):&=\|\phi^{\beta^m_{0}(\omega)}_{t_0}(\omega)\|\prod_{i=1}^k\|\psi^{\beta^m_i(\omega), (i-1)}_{t_i}(\omega)\|,& \quad &\omega \in \Omega.
  \end{align*}

\noindent
We make the following observations on this construction:
\begin{itemize}
\item[(i)]
The construction always produces a sequence $\varphi^m$ such that $\limsup_m \|\varphi^m_s\|<\infty$ for all $s \in \cT^*$.
Indeed, if $t_{i}=T$ for some $i=1, \dots, T$, then $\lim_m \psi^{m, (i)}_s(\omega)=0$ for  $s=1, \dots, T-1$, for a.e. $\omega \in \Omega$, and $\lim_m\|\psi^{m, (i)}_{T}(\omega)\|=1_{E^i}(\omega)$, for a.e. $\omega \in \Omega$.
The sequence $\psi^{m, (i)}$ clearly satisfies $\limsup_m \|\psi^{m, (i)}_s\|<\infty$ for all $s\in \cT^*$.
\item[(ii)]
 We have that $\varphi^{m}_s\in L^0(\Omega, \cF_{t_k-1}, \bP, \bR^N)$ for $s =1, \dots, t_k-1$, and\\ $\varphi^{m}_s\in L^0(\Omega, \cF_{s-1}, \bP, \bR^N)$ for $s=t_k, \dots, T$.
 Hence, the sequence $\varphi^{m}$ \emph{is not} a sequence of predictable processes.
However, the limit of any a.s. convergent subsequence of $\varphi^{m}$ \emph{is} predictable because $\varphi^m_s$ converges a.s. to 0 for $s=1, \dots, t_k-1$.
\item[(iii)]
$E^k \subseteq \cdots \subseteq E^0$, and $\bP(E^k) >0$.
\item[(iv)]
Any a.s. convergent subsequence of $\varphi^{m}$ converges a.s. to a nonzero process since $\|\varphi^{m}_{t_k}\|$  converges a.s. to $1_{E^k}$, which is nonzero a.s. since $\bP(E^k)>0$.
\item[(v)]
From  \eqref{eq:CloseNormalizedPhi} and \eqref{eq:CloseRecursion}, we have  $\varphi^{m}_s=1_{E}\phi^{\beta^m_{0}}_s/U^m$ for all $s \in \cT^*$,  where $E:= \bigcap_{i=1}^k E^i$.
Because $E^k \subseteq \cdots \subseteq E^0$,
 \begin{align}
  \varphi^{m}_s=1_{E^k}\frac{\phi^{\beta^m_{0}}_s}{U^m}, \quad s \in \cT^*. \label{eq:KClose14}
  \end{align}
 \item[(vi)]
$U^m(\omega)$ diverges for a.e. $\omega \in E^k$ since \eqref{eq:ClosePhiSupInfty}, \eqref{eq:CloseNestedTau}, and \eqref{eq:ClosePsiSupInfty} hold.
\end{itemize}

\noindent
\textbf{Step 1b:}
By the previous step, $\limsup_m \|\varphi^m_s\|<\infty$ for all $s \in \cT^*$.
We apply Lemma~\ref{Arbitrage/MeasurableSelectionShachProcess} to $\varphi^m$ to find a strictly increasing sequence of positive, integer-valued, $\cF_{T-1}$-measurable random variables $\rho^m$ so that $\varphi^{\rho^m}$ converges a.s. to some process $\varphi$ such that\footnote{See observation (ii) in Step 1a.} $\varphi_s \in L^0(\Omega, \cF_{t_k-1}, \bP; \bR^N)$ for $s=1, \dots, t_k-1$, and $\varphi_s \in L^0(\Omega, \cF_{s-1}, \bP; \bR^N)$ for $s=t_k, \dots, T$.
By  observation (ii) in Step 1a, we have that $\varphi$ is predictable.

\noindent
\textbf{Step 1c:}
We proceed by showing that \textbf{NAEF} implies $\bP(E^0)=0$.
Towards this, we first show that the process $\varphi$ constructed in Step 1b  satisfies $F(\varphi) \in \bK$.
For the sake of notation, we define  $\eta^m:=\beta^{\rho^m}_0$.
From \eqref{eq:KClose14}, we have  $\varphi^{\rho^m}=1_{E^k}\phi^{\eta^m}/U^{\rho^m}$.
Since $1_{E^k}$ and $U^{\rho^m}$ are nonnegative, $\bR$-valued random variables,
\begin{equation}\label{eq:ClosePosHomoF}
  1_{E^k}\frac{F( \phi^{\eta^m})}{U^{\rho^m}}=F\Big(1_{E^k}\frac {\phi^{\eta^m}}{U^{\rho^m}}\Big)=F(\varphi^{\rho^m}).
\end{equation}
Because $\varphi^{\rho^m}$ converges a.s. to $\varphi$, we may apply Lemma~\ref{Arbitrage/ValueConvergenceSF} to see that $F(\varphi^{\rho^m})$ converges a.s. to $F(\varphi)$.
Since $\varphi$ is predictable, we have from the definition of $\bK$ that $F(\varphi) \in \bK$.

We proceed by showing that $F(\varphi) \in L^0_+(\Omega, \cF_T, \bP; \bR)$.
Lets begin by defining $\widetilde X^m:=X^{\eta^m}/U^{\rho^m}$ and $\widetilde Z^m:=Z^{\eta^m}/U^{\rho^m}$.
From \eqref{eq:KClose11},
\begin{equation}\label{eq:KClose40}
  F(\phi^{\eta^m})=X^{\eta^m}+Z^{\eta^m}.
\end{equation}
By multiplying both sides of \eqref{eq:KClose40} by $1_{E^k}/U^{\rho^m}$, we see from \eqref{eq:ClosePosHomoF} that
\begin{equation}\label{eq:KClose27}
F(\varphi^{\rho^m})=1_{E^k}(\widetilde X^m+\widetilde Z^m).
\end{equation}
The sequence $X^m$ converges a.s. by assumption, so the sequence $X^{\eta^m}$ also converges a.s.
Recall that the sequence $U^{m}(\omega)$ diverges\footnote{See observations (vi)  in Step 1a.} for a.e. $\omega \in E^k$, so $U^{\rho^m}(\omega)$ diverges for a.e. $\omega \in E^k$ since $\{\rho^1(\omega), \rho^2(\omega), \dots, \}  \subseteq \bN$ for a.e. $\omega \in \Omega$.
Hence, $1_{E^k}\widetilde X^m$ converges a.s. to 0.
Since $F( \varphi^{\rho^m})$ and $1_{E^k}\widetilde X^m$ converge a.s.,  the sequence $1_{E^k}\widetilde Z^m$ also converges a.s. to some $Z \in L^0_+(\Omega, \cF_T, \bP; \bR)$.
Thus,  $F(\varphi^{\rho^m})$ converges a.s. to $Z$, which implies $F(\varphi) \in L^0_+(\Omega, \cF_T, \bP; \bR)$.

Since $F(\varphi) \in \bK$, we immediately see that $F(\varphi)\in \bK \cap L^0_+(\Omega, \cF_T, \bP; \bR)$.
It is assumed that \textbf{NA} is satisfied, so by Lemmas~\ref{Arbitrage/LemmaEquArb} and \ref{lemma: Kequivalent} we deduce that $F(\varphi)=0$.
We are supposing that \textbf{EF} holds, so according to Lemma~\ref{lemma: EFequivalent} we have $\varphi=0$.
This cannot happen given our assumption that $\bP(E^k)>0$ because\footnote{See observations (iv) in Step 1a.} $\|\varphi_{t_k}\|=1_{E^k}$.
Therefore, we must have that $\bP(E^k)=0$.
This contradicts the construction in Step 1a, so $\bP(E^0)=0$.

\noindent
\textbf{Step 2:}
By the conclusion in Step 1, we obtain that $\limsup_m \|\phi^m_s\| <\infty$ for $s \in \cT^*$.
By applying Lemma~\ref{Arbitrage/MeasurableSelectionShachProcess} to $\phi^m$, we may find a strictly increasing sequence of positive, integer-valued, $\cF_{T-1}$-measurable random variables $\sigma^m$ such that $\phi^{\sigma^m}$ converges a.s. to some predictable process $\phi$.

By Lemma~\ref{Arbitrage/ValueConvergenceSF}, the sequence $F(\phi^{\sigma^m})$ converges a.s. to $F(\phi)$.
Since $\phi \in \cP$, we have  $F(\phi) \in \bK$.
Because $X^{m}$ converges a.s. to $X$, the sequence $X^{\sigma^m}$ also converges a.s. to $X$.
From \eqref{eq:KClose11}, it is true that $X^{\sigma^m}=F(\phi^{\sigma^m})-Z^{\sigma^m}$.
Since $X^{\sigma^m}$ and $F(\phi^{\sigma^m})$ converges a.s., the sequence $Z^{\sigma^m}$ also converges a.s.
Thus,  $F(\phi^{\sigma^m})-X^{\sigma^m}$ converges a.s. to some nonnegative random variable $Z:=F(\phi)-X$, which gives us that $X=F(\phi)-Z$.
We conclude that $X \in \bK-L^0_+(\Omega, \cF_T, \bP; \bR)$.
\end{proof}

\subsection{The First Fundamental Theorem of Asset Pricing}\label{subsection:FFTAP}
In this section, we formulate and prove a version of the First Fundamental Theorem of Asset Pricing (FFTAP).
We  define the following set for convenience:
\begin{equation*}
  \cZ:=\{ \bQ: \ \bQ \sim \bP, \ P^{ask, *}, P^{bid, *}, A^{ask, *}, A^{bid, *} \ \textrm {are } \bQ\textrm{-integrable}\}.
\end{equation*}

We now define a risk-neutral measure in our context.

\begin{definition}\label{def:RiskNeutral}
 A probability measure $\bQ$  is a \emph{risk-neutral measure} if $\bQ \in \cZ$, and if $\mathbb{E}_{\bQ}[V^*_T(\phi)] \leq 0$ for all $\phi\in \cS$ such that $\phi^j$ is bounded a.s., for $j \in \cJ^*$, and $V_0(\phi)=0$.
\end{definition}

A natural question to ask is whether the expectation appearing in the definition above exists.
The following lemma shows that, indeed, it does.

\begin{lemma}\label{lemma:RNIntegrability}
Suppose that $\bQ \in \cZ$,  and let $\phi \in \cS$ be such that $\phi^j$ is bounded a.s., for $j \in \cJ^*$, and $V_0(\phi)=0$.
Then $V^*_T(\phi)$ is $\bQ$-integrable.
\end{lemma}

\begin{proof}
From the definition of $\cZ$, the processes $P^{ask, *}$, $P^{bid, *}$, $A^{ask, *}$, $A^{bid, *}$ are $\bQ$-integrable.
Because $\bQ$ is equivalent to $\bP$, and since $\phi^j$ is bounded $\bP$-a.s. for $j \in \cJ^*$, we have that $\phi^j$ is bounded $\bQ$-a.s. for $j \in \cJ^*$.
Therefore, we see from Proposition~\ref{Intro/LemmaDiscountedSelfFinancing} that $\bE_{\bQ}[|V^*_T(\phi)|] <\infty$ holds.
\end{proof}

\begin{remark}
For frictionless markets ($P:=P^{ask}=P^{bid}, A:=A^{ask}=A^{bid}$), a risk-neutral measure is classically defined to be an equivalent probability measure such that the discounted cumulative price process $(P_t+\sum_{u=1}^{t}\Delta A_u)_{u=0}^T$ is a martingale under $\bQ$.
    The present definition of a risk-neutral coincides with this classic definition of a risk-neutral measure if the market is frictionless.
     Indeed, if there are no frictions the value process satisfies $V^*_T(-\phi)=-V^*_T(\phi)$ for all trading strategies.
     Also, by Proposition~\ref{Intro/LemmaDiscountedSelfFinancing}, we have that
      $\bE_{\bQ}[V^*_T(\phi)]=0$ for all $\phi \in \cS$ such that $\phi^j$ is bounded a.s. and $V_0(\phi)=0$ for $j \in \cJ^*$ if and only if
      \begin{equation*}
      \sum_{j=1}^N\sum_{u=1}^{T}\bE_{\bQ}\Bigg[\phi^j_{u}\bE_{\bQ}\Big[\Delta \Big(B^{-1}_{u}P^j_{u} + \sum_{w=1}^u B^{-1}_w \Delta A^j_w\Big)\Big| \cF_{u-1}\Big]\Bigg]=0
      \end{equation*}
      for all $\phi \in \cS$ such that $\phi^j$ is bounded a.s., for $j \in \cJ^*$.
\end{remark}

\noindent
The next lemma provides a mathematically convenient condition that is equivalent to \textbf{NA}.

\begin{lemma}\label{lemma:NAbounded}
The no-arbitrage condition \textbf{\emph{(NA)}} is satisfied if and only if for each $\phi \in \cS$ such that $\phi^j$ is bounded a.s. for $j \in \cJ^*$,  $V_0(\phi)=0$, and $V_T(\phi) \geq 0$, we have  $V_T(\phi)=0$.
 \end{lemma}

\begin{proof}
Necessity holds immediately, so we only show sufficiency.
Let $\phi\in \cS$ be a trading strategy so that $V_0(\phi)=0$ and $V_T(\phi) \geq 0$.
We will show that $V_T(\phi)=0$.

First, define the $\cF_{t-1}$ measurable set $\Omega^{m, j}_t:=\{\omega \in \Omega: |\phi^j_t(\omega)| \leq m\}$ for $m \in \bN$, $t \in \cT^*$, and $j \in \cJ^*$, and define the sequence of trading strategies $\psi^m$ as $\psi^{m, j}_t:=1_{\Omega^{m, j}_t}\phi^j_t$ for $t \in \cT^*$ and $j \in \cJ^*$, where $\psi^{m,0}$ is chosen so that $\psi^m$ is self-financing and $V_0(\psi^m)=0$.
Since $1_{\Omega^{m,j}_t}$ converges a.s. to $1$ for all $t \in \cT^*$ and $j \in \cJ^*$, we have that $\psi^{m, j}_t$ converges a.s. to $\phi^j_t$ for all $t \in \cT^*$ and $j \in \cJ^*$.

Now we prove that $V_0(\psi^m)$ converges a.s. to $V_0(\phi)$.
Towards this, we first show that $\psi^{m, j}_1$ converges a.s. to $\phi^j_1$ for all $j \in \cJ$.
By the definition of $V_0(\psi^m)$,
\begin{equation*}
  V_0(\psi^m)= \psi^{m, 0}_1+\sum_{j=1}^N \psi^{m, j}_1\big(1_{\{\psi^{m, j}_1 \geq 0\}}P^{ask, j}_0+1_{\{\psi^{m, j}_1 <0 \}}P^{bid, j}_0 \big).
\end{equation*}
Since $\psi^{m, 0}$ is chosen so that $V_0(\psi^m)=0$, we have
\begin{equation*}
  \psi^{m, 0}_1=-\sum_{j=1}^N \psi^{m, j}_1\big(1_{\{\psi^{m, j}_1 \geq 0\}}P^{ask, j}_0+1_{\{\psi^{m, j}_1 <0 \}}P^{bid, j}_0 \big).
\end{equation*}
The sequence $\psi^{m, j}_1$ converges a.s. to $\phi^j_1$ for all $j \in \cJ^*$, so by Lemma~\ref{Arbitrage/ValueConvergenceSF} the sequence $\psi^{m, 0}_1$ converges a.s. to
\begin{equation*}
  -\sum_{j=1}^N \phi^{ j}_1\big(1_{\{\phi^{ j}_1 \geq 0\}}P^{ask, j}_0+1_{\{\phi^{j}_1 <0 \}}P^{bid, j}_0 \big).
\end{equation*}
However, $V_0(\phi)=0$, so
$\psi^{m, 0}_1$ converges a.s. to $\phi^0_1$
Thus, $\psi^{m, j}_1$ converges a.s. to $\phi^j_1$ for all $j \in \cJ$.
By Lemma~\ref{Arbitrage/ValueConvergenceSF}, we have that $V_0(\psi^m)$ converges to $V_0(\phi)$.

According to Lemma~\ref{Arbitrage/ValueConvergenceSF},  $V^*_T(\psi^{m})$ converges a.s. to $V^*_T(\phi)$ since $\psi^{m, j}_t$ converges a.s. to $\phi^j_t$ for all $t \in \cT^*$ and $j \in \cJ^*$.

Next, since $\psi^m$ is self-financing and $\psi^{m,j}$ is bounded a.s. for all $j \in \cJ^*$ and all $m \in \bN$, we obtain
\begin{equation*}
V_0(\psi^m)=0, \  V^*_T(\psi^m) \geq 0 \ \Longrightarrow \ V^*_T(\psi^m) =0, \quad m \in \bN.
\end{equation*}
Since $V_0(\psi^m)$ converges a.s. to $V_0(\phi)$, and $V^*_T(\psi^m)$ converges a.s. to $V^*_T(\phi)$ we have
\begin{equation*}
V_0(\phi)=0, \   V_T(\phi) \geq 0 \ \Longrightarrow \ V_T(\phi) =0.
\end{equation*}
Since $V_0(\phi)=0$ and $V_T(\phi) \geq 0$, we conclude that $V_T(\phi)=0$, so  \textbf{NA} holds.
\end{proof}

Next, we recall the well-known  Kreps-Yan Theorem.
It was first proved by Yan~\cite{Yan1980}, and then obtained independently by Kreps~\cite{Kreps1981} in the context of financial mathematics.
For a proof of the version presented in this paper, see Schachermayer~\cite{Schachermayer1992}.
Theorem~\ref{lemma:KClosed} and the Kreps-Yan Theorem will essentially imply the FFTAP (Theorem~\ref{Theorem: FFTAP}).

\begin{theorem}[Kreps-Yan]\label{Arbitrage/LemmaYanKreps}
Let $\mathcal{C}$ be a closed convex cone in $L^1(\Omega, \cF, \bP; \bR)$ containing $L^1_{-}(\Omega, \cF, \bP; \bR)$ such that $\mathcal{C} \cap L^1_+(\Omega, \cF, \bP; \bR)=\{0\}$.
Then there exists a functional $f \in L^{\infty}(\Omega, \cF, \bP; \bR)$ such that, for each $h \in L^1_+(\Omega, \cF, \bP; \bR)$ with $h \neq 0$, we have that $\bE_{\bP}[fh]>0$ and $\bE_{\bP}[fg] \leq 0$ for any $g \in \mathcal{C}$.
\end{theorem}

\noindent
We are now ready prove the following version of the FFTAP.

\begin{theorem}[First Fundamental Theorem of Asset Pricing]\label{Theorem: FFTAP}
The following conditions are equivalent:
\begin{itemize}
\item[(i)] The no-arbitrage condition under the efficient friction assumption \emph{\textbf{(NAEF)}} is satisfied.
\item[(ii)] There exists a risk-neutral measure.
\item[(iii)] There exists a risk-neutral measure $\bQ$ so that $\mathrm{d}\bQ/\mathrm{d}\bP \in L^{\infty}(\Omega, \cF_T, \bP; \bR)$.
\end{itemize}
\end{theorem}

\begin{proof}
In order to prove these equivalences, we show that $(ii) \Rightarrow (i)$, $(i) \Rightarrow (iii)$, and $(iii) \Rightarrow (ii)$.
The implication $(iii) \Rightarrow (ii)$ is immediate, so we only show the remaining two.
\smallskip

$(ii) \Rightarrow (i)$:
We prove by contradiction.
Assume there exists a risk-neutral measure $\bQ$, and that \textbf{NA} does not hold.
By Lemma~\ref{lemma:NAbounded}, there exists $\phi \in \cS$ so that $\phi^j$ is bounded a.s., for $j \in \cJ^*$,  $V_0(\phi)=0$, $V^*_T(\phi) \geq 0$, and $\bP(V^*_T(\phi)(\omega) >0)>0$.
Since $\bQ$ is equivalent to $\bP$, we have $V_0(\phi)=0$, $V^*_T(\phi) \geq 0$ $\bQ$-a.s., and $\bQ(V^*_T(\phi)(\omega) >0)>0$.
So $\bE_{\bQ}[V^*_T(\phi)] > 0$, which contradicts that $\bQ$ is risk-neutral.
Hence, \textbf{NA} holds.

$(i) \Rightarrow (iii)$:
We first construct a probability measure $\widetilde \bP$ satisfying  $\widetilde \bP \in \cZ$ and $\d \widetilde{\bP} / \d \bP \in L^{\infty}(\Omega, \cF_T, \bP; \bR)$.
Towards this, let us define the $\cF_T$-measurable weight function
\begin{equation}\label{eq:FFAPWeight}
w:=1+ \sum_{u=0}^T\|P^{ask, *}_u\|+\sum_{u=0}^T\|P^{bid, *}_u\|+ \sum_{u=1}^T\|A^{ask, *}_u\|+\sum_{u=1}^T\|A^{bid, *}_u\|,
\end{equation}
and let $\widetilde \bP$ be the measure on $\cF_T$  with Radon-Nikod\'{y}m derivative $\d \widetilde{\bP} / \d \bP=\tilde{c}/w$, where $\widetilde c$ is an appropriate normalizing constant.
We see that $\widetilde \bP$ is equivalent to $\bP$, and $\d \widetilde{\bP} / \d \bP \in L^{\infty}(\Omega, \cF_T, \bP; \bR)$.
By the choice of the weight function $w$, the processes $P^{ask, *}$, $P^{bid, *}$, $A^{ask, *}$, $A^{bid, *}$ are $\widetilde{\bP}$-integrable.
Thus, $\widetilde \bP \in \cZ$.

Next, since $\widetilde \bP$ is equivalent to $\bP$, it follows that
$\big(\cK-L^0_+(\Omega, \cF_T, \widetilde\bP\,; \bR)\big) \cap L^0_+(\Omega, \cF_T, \widetilde{\bP}\,; \bR)=\{0\}$ by Lemma~\ref{Arbitrage/LemmaEquArb}, the set $\cK-L^0_+(\Omega, \cF_T, \widetilde\bP\,; \bR)$ is $\widetilde{\bP}$-closed  according to Theorem~\ref{lemma:KClosed}, and  $\cK-L^0_+(\Omega, \cF_T, \widetilde \bP\,; \bR)$ is a convex cone by Lemma~\ref{Arbitrage/LemmaConvexCone}.

Let us now consider the set $\cC:=\big(\cK-L^0_+(\Omega, \cF_T, \widetilde \bP\,; \bR)\big) \cap L^1(\Omega, \cF_T, \widetilde{\bP}\,;\bR)$.
We observe the following that $\cC\cap L^1_+(\Omega, \cF_T, \widetilde\bP\,; \bR)=\{0\}$, $\cC$ is a convex cone, and $\cC \supseteq L^1_-(\Omega, \cF_T, \widetilde{\bP}\,;\bR)$.
Moreover, since convergence in $L^1(\Omega, \cF_T, \widetilde \bP\,; \bR)$ implies convergence in measure $\widetilde{\bP}$,  we have that $\cC$ is closed in $L^1(\Omega, \cF_T, \widetilde \bP\,; \bR)$.

Thus, according to Theorem~\ref{Arbitrage/LemmaYanKreps},  there exists a strictly positive functional\footnote{For each $h \in L^1_+(\Omega, \cF_T, \widetilde\bP\,; \bR)$ with $h \neq 0$, we have $\mathbb{E}_{\widetilde\bP}[fh]>0$.}\\ $f \in L^{\infty}(\Omega, \cF_T, \widetilde\bP\,; \bR)$ such that $\bE_{\widetilde\bP}[ Kf] \leq 0$ for all  $K \in \cC$.
Because $0 \in L^0_+(\Omega, \cF_T, \widetilde\bP\,; \bR)$, it follows from the definition of $\cC$ that
\begin{equation*}
\bE_{\widetilde\bP}[ Kf] \leq 0, \quad  K \in \cK \cap L^1(\Omega, \cF_T, \widetilde\bP\,; \bR).
\end{equation*}
By the definition of $\cK$, this implies that $\bE_{\widetilde\bP}[ V^*_T(\phi)f] \leq 0$ for all $\phi \in \cS$ such that $V_0(\phi)=0$ and $V^*_T(\phi) \in L^1(\Omega, \cF_T, \widetilde\bP\,; \bR)$.
In particular,  $\bE_{\widetilde\bP}[ V^*_T(\phi)f] \leq 0$ for all $\phi \in \cS$ such that $\phi^j$ is bounded a.s., for  $j \in \cJ^*$, $V_0(\phi)=0$, and $V^*_T(\phi) \in  L^1(\Omega, \cF_T, \widetilde\bP\,; \bR)$.
 Since $\widetilde{\bP} \in \cZ$, we obtain from Lemma~\ref{lemma:RNIntegrability} that $V^*_T(\phi)$ is $\widetilde \bP$-integrable.
 Thus,  $\bE_{\widetilde\bP}[ V^*_T(\phi)f] \leq 0$ for all $\phi \in \cS$ such that $\phi^j$ is bounded a.s., for  $j \in \cJ^*$, and $V_0(\phi)=0$.

We proceed by constructing a risk-neutral measure.
Let $\bQ$ be the measure on $\cF_T$ with Radon-Nikod\'{y}m derivative $\mathrm{d}\bQ/\mathrm{d} \widetilde{\bP}:= c f$,
 where $c$ is an appropriate normalizing constant.
Because $f$ is a strictly positive functional in $L^{\infty}(\Omega, \cF_T, \widetilde\bP\,; \bR)$, we have that $\bQ$ is equivalent to $\widetilde \bP$.
Since $\widetilde \bP$ is equivalent to $\bP$, it follows that $\bQ$ is equivalent to $\bP$.
Also,
\begin{equation}\label{eq:FFAP1}
  \frac{\d \bQ}{\d \bP}=\frac{\d \bQ}{\d \widetilde {\bP}}\frac{\d \widetilde {\bP}}{\d \bP} =c \tilde{c} \frac{f}{w}.
\end{equation}
Thus, since $w \geq 1$ and $f \in L^{\infty}(\Omega, \cF_T, \widetilde\bP\,; \bR)$, we have $\d \bQ/ \d \bP \in L^{\infty} (\Omega, \cF_T, \widetilde \bP; \bR)$.
This gives us $\d \bQ/ \d \bP \in L^{\infty} (\Omega, \cF_T, \bP; \bR)$ since $\widetilde \bP$ is equivalent to $\bP$.
Moreover, we note that the processes $P^{ask, *}$, $P^{bid,*}$, $A^{ask, *}$, $A^{bid,*}$ are $\bQ$-integrable.
Hence, $\bQ \in \cZ$.
We conclude that  $\bQ$ is a risk-neutral measure since $\bE_{\bQ}[V^*_T(\phi)] =c \:\bE_{\widetilde \bP}[V^*_T(\phi)f] \leq 0$ for all $\phi \in \cS$ such that $\phi^j$ is bounded a.s., for  $j \in \cJ^*$, and $V_0(\phi)=0$.
\end{proof}

\begin{remark}\mbox{}
\begin{itemize}
\item[(i)] Note that \textbf{EF} is not needed to prove the implication $(ii) \Rightarrow (i)$.
\item[(ii)] In practice, it is typically required for a market model to satisfy \textbf{NA}.
    According to Theorem~\ref{Theorem: FFTAP}, it is enough to check that there exists a risk-neutral measure.
    However, this is not straightforward because it has to be verified whether there exists a probability measure $\bQ \in \cZ$ so that $\bE_\bQ[V^*_T(\phi)] \leq 0$ \emph{for all} $\phi \in \cS$ so that $\phi^j$ is bounded a.s., for $j \in \cJ^*$, and $V_0(\phi)=0$.
   We will show in the following section that consistent pricing systems help solve this issue (see Proposition~\ref{proposition:ConsistentRN} and Theorem~\ref{theorem:ConsistentRNFrictionless}).
\end{itemize}
\end{remark}

\section{Consistent pricing systems}\label{section:CPS}
Consistent pricing systems (CPSs) are instrumental in  the theory of arbitrage in markets with transaction costs---they provide a bridge between martingale theory in the theory of arbitrage in frictionless markets and more general concepts in theory of arbitrage in markets with transaction costs.
Essentially, CPSs are interpreted as corresponding auxiliary frictionless markets.
They are very useful from the practical point of view because they provide a straightforward way to verify whether a financial market model satisfies \textbf{NA}.
In this section, we explore the relationship between CPSs and \textbf{NA}.

We begin by defining a CPS in our context.

\begin{definition}
A \emph{consistent pricing system} (CPS) corresponding to the market \\$(B, P^{ask}, P^{bid}, A^{ask}, A^{bid})$  is a quadruplet $\{\bQ, P, A, M\}$ consisting of
\begin{itemize}
  \item[(i)]
  a probability measure $\bQ\in \cZ$;
  \item[(ii)]
  an adapted process $P$ satisfying $P^{bid, *} \leq P \leq P^{ask, *}$;
  \item[(iii)]
  an adapted process $A$ satisfying $A^{ask, *} \leq A \leq A^{bid, *}$;
  \item[(iv)]
 a martingale $M$ under $\bQ$ satisfying $M_t=P_t+\sum_{u=1}^t A_u$ for all $t \in \cT$.
\end{itemize}
\end{definition}

\begin{remark} Since our market is fixed throughout the paper, we shall simply refer to  $\{\bQ, P, A, M\}$  as a CPS, rather than a CPS corresponding to the market\\ $(B, P^{ask}, P^{bid}, A^{ask}, A^{bid})$.
\end{remark}

 For a CPS  $\{\bQ, P, A, M\}$, the process $P$ is interpreted as the corresponding auxiliary  frictionless ex-dividend price process, and the process $A$ has the interpretation of the corresponding auxiliary  frictionless cumulative dividend process,
whereas $M$ is viewed as the corresponding auxiliary frictionless cumulative price process.

The next result establishes a relationship between \textbf{NA} and CPSs in our context.

\begin{proposition}\label{proposition:ConsistentRN}
  If there exists a consistent pricing system \emph{(CPS)}, then the no-arbitrage condition \emph{\textbf{(NA)}} is satisfied.
\end{proposition}

\begin{proof}
Suppose there exists a CPS, call it $\{\bQ, P, A, M\}$, and suppose $\phi\in \cS$ is a trading strategy such that $\phi^j$ is bounded a.s., for $j \in \cJ^*$, and $V_0(\phi)=0$.
In view of Proposition~\ref{Intro/LemmaDiscountedSelfFinancing}, and because $P^{bid} \leq P \leq  P^{ask}$ and $A^{ask} \leq A \leq A^{bid}$, we deduce that
\begin{equation*}
  V^*_T(\phi) \leq \sum_{j=1}^N\Big(\phi^{j}_TP^{ j}_T+\sum_{u=1}^T(-\Delta\phi^j_uP^{j}_{u-1}+\phi^{ j}_u A^{j}_u)\Big).
\end{equation*}
Since $M=P+\sum_{u=1}^{\cdot} A_u$ is a martingale under $\bQ$, and because $\phi^j$ is bounded a.s., for $j \in \cJ^*$, we have
\begin{align*}
\mathbb{E}_{\bQ}[V^*_T(\phi)]&\leq\sum_{j=1}^N \mathbb{E}_{\bQ}\Big[\phi^{j}_TP^{ j}_T+\sum_{u=1}^T(-\Delta\phi^j_uP^{j}_{u-1}+\phi^{ j}_u A^{j}_u)\Big]\\
&=\sum_{j=1}^N \sum_{u=1}^T\mathbb{E}_{\bQ}\Bigg[\Delta\phi^j_u \mathbb{E}_{\bQ}\Big[P^{j}_T+\sum_{w=1}^T A^{j}_w-P^{j}_{u-1}-\sum_{w=1}^{u-1} A^{j}_w\Big| \cF_{u-1}\Big]\Bigg]\\
&=\sum_{j=1}^N \sum_{u=1}^T\mathbb{E}_{\bQ}\Bigg[\Delta\phi^j_u \mathbb{E}_{\bQ}\Big[M^j_T-M^j_{u-1}\Big| \cF_{u-1}\Big]\Bigg]\\
&=0.
\end{align*}
Therefore $\bQ$ is a risk-neutral measure.
According to Theorem~\ref{Theorem: FFTAP}, \textbf{NA} holds.
\end{proof}

At this point, a natural question to ask is whether there exists a CPS whenever \textbf{NA} is satisfied.
In general, this is still an open question.
However, for the special case in which there are no transaction costs in the dividends paid by the securities, $A^{ask}=A^{bid}$, will show in Theorem~\ref{theorem:ConsistentRNFrictionless} that there exists a CPS if and only if \textbf{NAEF} is satisfied.

 Proposition \ref{proposition:ConsistentRN} is important from the modeling point of view because it provides a sufficient condition for a model to satisfy \textbf{NA}.
  In the next example, we construct a model for which there exists a CPS.

\begin{example}\label{example:CDSConsistent}
Lets consider the CDS specified in Example \ref{example:CDS}.
Recall that  the cumulative dividend processes $A^{ask}$ and $A^{bid}$ corresponding to the CDS are defined as
\begin{align*}
&  A^{ask}_t:=1_{\{\tau \leq t\}} \delta -\kappa^{ask} \sum_{u=1}^t 1_{\{u <\tau\}}, \quad  A^{bid}_t:=1_{\{\tau \leq t\}} \delta -\kappa^{bid} \sum_{u=1}^t 1_{\{u <\tau\}}
\end{align*}
for all $t \in \cT^*$.
Let us fix any probability measure $\bQ$ equivalent to $\bP$.
We postulate that the ex-dividend prices $P^{ask}$ and $P^{bid}$ satisfy
\begin{align*}
  P^{ask, *}_t&=\bE_\bQ \Big[\sum_{u=t+1}^T A^{bid, *}_u \Big| \cF_t\Big], \quad  P^{bid, *}_t=\bE_\bQ \Big[\sum_{u=t+1}^T A^{ask, *}_u \Big| \cF_t\Big],
\end{align*}
for all $t \in \cT^*$.
By substituting $A^{ask, *}$ and $A^{bid, *}$ into the equations for $P^{ask, *}$ and $P^{bid, *}$ above, we see that
\begin{align*}
  P^{ask, *}_t&=\bE_\bQ \Big[ 1_{\{t<\tau \leq T\}}B^{-1}_{\tau}\delta-\kappa^{bid} \sum_{u=t+1}^T B^{-1}_u 1_{\{u <\tau\}}\Big| \cF_t\Big],\\
  P^{bid, *}_t&=\bE_\bQ \Big[ 1_{\{t<\tau \leq T\}}B^{-1}_{\tau}\delta-\kappa^{ask} \sum_{u=t+1}^T B^{-1}_u 1_{\{u <\tau\}}\Big| \cF_t\Big].
\end{align*}
For a fixed  $ \kappa \in [\kappa^{bid}, \kappa^{ask}]$, we define
  \begin{align*}
    A_t&:=B^{-1}_{t}\big(1_{\{\tau =t\}}\delta-\kappa  1_{\{t <\tau\}}\big), &\quad & t \in \cT^*,\\
    P_t&:=\bE_\bQ \Big[\sum_{u=t+1}^T A_u \Big| \cF_t\Big]=\bE_\bQ \Big[ 1_{\{t<\tau \leq T\}}B^{-1}_{\tau}\delta-\kappa \sum_{u=t+1}^T B^{-1}_u 1_{\{u <\tau\}}\Big| \cF_t\Big], & \quad  & t \in \cT,\\
    M_t&:=P_t+\sum_{u=1}^t A_u, & \quad& t \in \cT.
  \end{align*}
The quadruplet $\{\bQ, P, A, M\}$ is a CPS.
 To see this, first observe that $A$ and $P$ are $\bQ$-integrable since $A$ is bounded $\bQ$-a.s.
 Thus, $\bQ \in \cZ$.
  Next,  $M$ satisfies
  \begin{equation*}
    M_t=\bE_\bQ \Big[ 1_{\{\tau \leq T\}}B^{-1}_{\tau}\delta-\kappa \sum_{u=1}^T B^{-1}_u 1_{\{u <\tau\}}\Big| \cF_t\Big], \quad t \in \cT,
  \end{equation*}
so $M$ is  a Doob martingale under $\bQ$.
 Also, since $ \kappa \in [\kappa^{bid}, \kappa^{ask}]$, we have $A^{ask, *} \leq A \leq A^{bid, *}$ and $P^{bid, *} \leq P \leq P^{ask, *}$.
  Thus, $\{\bQ, P, A, M\}$ is a CPS.
 According to Proposition \ref{proposition:ConsistentRN}, we may additionally conclude that the financial market model $\{B, P^{ask}, P^{bid}, A^{ask}, A^{bid}\}$ satisfies \textbf{NA}.
\end{example}

\subsection{Consistent pricing systems under the assumption \texorpdfstring{ $A^{ask}=A^{bid}$}{}}\label{subsection:CPSNoT}
In this section we investigate the relationship between risk-neutral measures and CPSs under the assumption $A^{ask}=A^{bid}$.
Let us denote by $A$ the process $A^{ask}$.
We begin by proving two preliminary lemmas that hold in general (without the assumption $A^{ask} =A^{bid}$).

\begin{lemma}\label{RNInequalitiesST}
  If $\bQ$ is a risk-neutral measure, then
    \begin{align*}
  & P^{bid, j, *}_{\sigma_1} \leq \mathbb{E}_{\bQ}\Big[P^{ask, j, *}_{\sigma_2} +\sum_{u=\sigma_1+1}^{\sigma_2}A^{bid, j, *}_u\Big | \cF_{\sigma_1}\Big],\ \  P^{ask, j, *}_{\sigma_1} \geq \mathbb{E}_{\bQ}\Big[P^{bid, j, *}_{\sigma_2} +\sum_{u=\sigma_1+1}^{\sigma_2}  A^{ask, j, *}_u\Big | \cF_{\sigma_1}\Big],
\end{align*}
 for all $j \in \cJ^*$ and stopping times $0 \leq \sigma_1 < \sigma_2 \leq T$.
\end{lemma}

\begin{proof}
Suppose $\bQ$ is a risk-neutral measure.
For stopping times $0 \leq \sigma_1 < \sigma_2 \leq T$ and random variables $\xi_{\sigma_1}\in L^{\infty}(\Omega, \cF_{\sigma_1}, \bP; \; \bR^{N})$, we define the trading strategy
\begin{equation*}
\theta(\sigma_1, \sigma_2, \xi_{\sigma_1}):=(\theta^0_t(\sigma_1, \sigma_2, \xi_{\sigma_1}), 1_{\{\sigma_1 +1\leq t \leq \sigma_2\}}\xi_{\sigma_1}^1, \dots, 1_{\{\sigma_1 +1\leq t \leq \sigma_2\}}\xi_{\sigma_1}^N)_{t=1}^T,
 \end{equation*}
 where $\theta^0(\sigma_1, \sigma_2, \xi_{\sigma_1})$ is chosen such that $\theta(\sigma_1, \sigma_2, \xi_{\sigma_1})$ is self-financing and\\ $V_0(\theta(\sigma_1, \sigma_2, \xi_{\sigma_1}))=0$.
Due to Proposition~\ref{Intro/LemmaDiscountedSelfFinancing}, the value process associated with $\theta$ satisfies
\begin{align*}
V^*_T(\theta(\sigma_1, \sigma_2, \xi_{\sigma_1}))&= \sum_{j=1}^N 1_{\{\xi_{\sigma_1}^j \geq 0\}}\xi_{\sigma_1}^j \Big(P^{bid, j, *}_{\sigma_2} +\sum_{u=\sigma_1+1}^{\sigma_2} A^{ask, j, *}_u  -P^{ask, j, *}_{\sigma_1}\Big)\\
  & \qquad +\sum_{j=1}^N 1_{\{\xi_{\sigma_1}^j < 0\}}\xi_{\sigma_1}^j\Big(P^{ask, j, *}_{\sigma_2} +\sum_{u=\sigma_1+1}^{\sigma_2} A^{bid, j, *}_u
  -P^{bid, j, *}_{\sigma_1}\Big).
\end{align*}
Since $\bQ$ is a risk-neutral measure,  we have $\bE_\bQ[V^*_T(\theta(\sigma_1, \sigma_2, \xi_{\sigma_1}))] \leq 0$ for all stopping times $0 \leq \sigma_1 < \sigma_2 \leq T$ and $\xi_{\sigma_1} \in L^{\infty}(\Omega, \cF_{\sigma_1}, \bP; \; \bR^{N})$.
Hence, we are able to obtain
   \begin{align*}
&\bE_\bQ\Big[\sum_{j=1}^N 1_{\{\xi_{\sigma_1}^j \geq 0\}}\xi_{\sigma_1}^j \Big(P^{bid, j, *}_{\sigma_2} +\sum_{u=\sigma_1+1}^{\sigma_2} A^{ask, j, *}_u  -P^{ask, j, *}_{\sigma_1}\Big)\\
  & \qquad +\sum_{j=1}^N 1_{\{\xi_{\sigma_1}^j < 0\}}\xi_{\sigma_1}^j\Big(P^{ask, j, *}_{\sigma_2} +\sum_{u=\sigma_1+1}^{\sigma_2} A^{bid, j, *}_u
  -P^{bid, j, *}_{\sigma_1}\Big)\Big] \leq 0,
  \end{align*}
 for  all stopping times $0 \leq \sigma_1 < \sigma_2 \leq T$ and random variables $\xi_{\sigma_1} \in L^{\infty}(\Omega, \cF_{\sigma_1}, \bP; \; \bR^{N})$.
By the tower property of conditional expectations, we get that
   \begin{align*}
&\bE_\bQ\Bigg[\sum_{j=1}^N 1_{\{\xi_{\sigma_1}^j \geq 0\}}\xi_{\sigma_1}^j \bE_\bQ\Big[P^{bid, j, *}_{\sigma_2} +\sum_{u=\sigma_1+1}^{\sigma_2} A^{ask, j, *}_u
  -P^{ask, j, *}_{\sigma_1}\Big| \cF_{\sigma_1}\Big]\\
  & \qquad +\sum_{j=1}^N 1_{\{\xi_{\sigma_1}^j < 0\}}\xi_{\sigma_1}^j\bE_\bQ\Big[P^{ask, j, *}_{\sigma_2} +\sum_{u=\sigma_1+1}^{\sigma_2} A^{bid, j, *}_u
  -P^{bid, j, *}_{\sigma_1}\Big| \cF_{\sigma_1} \Big]\Bigg] \leq 0.
  \end{align*}
  for all stopping times $0 \leq \sigma_1 < \sigma_2 \leq T$ and random variables $\xi_{\sigma_1} \in L^{\infty}(\Omega, \cF_{\sigma_1}, \bP; \; \bR^{N})$.
  This implies that the claim is satisfied.
\end{proof}

The next result is motivated by Theorem~4.5 in Cherny~\cite{Cherny2007a}.
We will denote by $\cT_t$  the set of stopping times in $\{t, t+1, \dots, T\}$, for all $t \in \cT$.
\begin{lemma}\label{lemma: Inequalities1}
Suppose $\bQ$ is a risk-neutral measure, and let
\begin{align*}
  &X^{b, j}_s:= \underset{\sigma \in \cT_s}{\emph{\esssup}}\;\mathbb{E}_{\bQ}\Big[P^{bid, j, *}_\sigma+ \sum_{u=1}^{\sigma}A^{ask, j, *}_u \Big | \cF_s\Big],\\
  &X^{a, j}_s:= \underset{\sigma\in \cT_s}{\emph{\essinf}}\;\mathbb{E}_{\bQ}\Big[P^{ask, j, *}_\sigma+ \sum_{u=1}^{\sigma}A^{bid, j, *}_u \Big | \cF_s\Big],
\end{align*}
for all $j \in \cJ^*$ and $s \in \cT$.
Then $X^{b}$ is a supermartingale and $X^a$ is a submartingale, both under $\bQ$, and satisfy $X^b \leq X^a$.
\end{lemma}

\begin{proof}
Let us fix $j \in \cJ^*$.
The processes $X^{b,j}$ and $X^{a, j}$ are Snell envelopes, so $X^{a, j}$ is a supermartingale and $X^{b. j}$ is a submartingale, both under $\bQ$ (see for instance, F\"{o}llmer and Schied~\cite{FollmerSchiedBook2004}).

We now show that $X^{b, j}\leq X^{a, j}$.
Let us define the process
\begin{equation*}
  X^j_t:=\mathbb{E}_{\bQ}\Big[P^{bid, j, *}_{\tau_1}+\sum_{u=1}^{\tau_1}A^{ask, j, *}_u\Big| \cF_t\Big] -\mathbb{E}_{\bQ}\Big[P^{ask, j, *}_{\tau_2}+\sum_{u=1}^{\tau_2}A^{bid, j, *}_u\Big| \cF_t\Big], \quad t \in \cT.
\end{equation*}
For any stopping times $\tau_1, \tau_2 \in \cT_t$, we see that
\begin{align*}
X^j_t& = \mathbb{E}_{\bQ}\Big[\mathbb{E}_{\bQ}\Big[P^{bid, j, *}_{\tau_1}+\sum_{u=1}^{\tau_1}A^{ask, j, *}_u-P^{ask, j, *}_{\tau_2}-\sum_{u=1}^{\tau_2} A^{bid, j, *}_u\big| \cF_{\tau_1 \wedge \tau_2} \Big] \Big| \cF_t\Big] \\
  & =\mathbb{E}_{\bQ}\Big[1_{\{\tau_1 \leq \tau_2\}}\Big(P^{bid, j, *}_{\tau_1}+\sum_{u=1}^{\tau_1}A^{ ask, j, *}_u-\mathbb{E}_{\bQ}\Big[P^{ask, j, *}_{\tau_2} +\sum_{u=1}^{\tau_2} A^{bid, j, *}_u \Big| \cF_{\tau_1}\Big]\Big) \Big| \cF_t \Big] \\
& + \mathbb{E}_{\bQ}\Big[1_{\{\tau_1 > \tau_2\}}\Big(\mathbb{E}_{\bQ}\Big[P^{bid, j, *}_{\tau_1}+\sum_{u=1}^{\tau_1}A^{ask, j, *}_u \Big| \cF_{\tau_2}\Big]-P^{ask, j, *}_{\tau_2}-\sum_{u=1}^{\tau_2} A^{bid, j, *}_u \Big) \Big| \cF_t \Big].
\end{align*}
After rearranging terms, we deduce that
\begin{align*}
 X^j_t&=\mathbb{E}_{\bQ}\Big[1_{\{\tau_1 \leq \tau_2\}}P^{bid, j, *}_{\tau_1}+\sum_{u=1}^{\tau_1}(A^{ ask, j, *}_u-A^{bid, j, *}_u)\\
&-1_{\{\tau_1 \leq \tau_2\}}\mathbb{E}_{\bQ}\Big[P^{ask, j, *}_{\tau_2} +\sum_{u=\tau_1+1}^{\tau_2} A^{bid, j, *}_u \Big| \cF_{\tau_1}\Big] \Big| \cF_t \Big] \\
& + \mathbb{E}_{\bQ}\Big[1_{\{\tau_1 > \tau_2\}}\mathbb{E}_{\bQ}\Big[P^{bid, j, *}_{\tau_1}+\sum_{u=\tau_2+1}^{\tau_1}A^{ask, j, *}_u \Big| \cF_{\tau_2}\Big]\\
&-1_{\{\tau_1 > \tau_2\}}(P^{ask, j, *}_{\tau_2}+\sum_{u=1}^{\tau_2} (A^{bid, j, *}_u-A^{ask, j, *}_u) ) \Big| \cF_t \Big].
\end{align*}
Because $A^{ask, *} \leq A^{bid, *}$, we are able to obtain
\begin{align}
  X^j_t& \leq \mathbb{E}_{\bQ}\Big[1_{\{\tau_1 \leq \tau_2\}}\Big(P^{bid, j, *}_{\tau_1}-\mathbb{E}_{\bQ}\Big[P^{ask, j, *}_{\tau_2} +\sum_{u=\tau_1+1}^{\tau_2} A^{bid, j, *}_u \Big| \cF_{\tau_1}\Big]\Big) \Big| \cF_t \Big] \notag\\
& \quad + \mathbb{E}_{\bQ}\Big[1_{\{\tau_1 > \tau_2\}}\Big(\mathbb{E}_{\bQ}\Big[P^{bid, j, *}_{\tau_1}+\sum_{u=\tau_2+1}^{\tau_1}A^{ask, j, *}_u \Big| \cF_{\tau_2}\Big]-P^{ask, j, *}_{\tau_2} \Big) \Big| \cF_t \Big]. \label{eq:Snell1}
\end{align}
Since $\bQ$ is a risk-neutral measure, we see from  Lemma~\ref{RNInequalitiesST} and \eqref{eq:Snell1} that $X^j_t \leq 0$.
The stopping times $\tau_1$ and $\tau_2$ are arbitrary in the definition of $X^j$, so we conclude that $X^{b, j} \leq X^{a, j}$.
\end{proof}

The next theorem gives sufficient and necessary conditions for there to exist a CPS (cf. Cherny~\cite{Cherny2007a}; Kabanov et al.~\cite{KabanovRasonyiStricker2002}; Schachermayer~\cite{Schachermayer2004}).

\begin{theorem}\label{theorem:ConsistentRNFrictionless}
Under the assumption that $A^{ask}=A^{bid}$, there exists a consistent pricing system \emph{(CPS)} if and only if the no-arbitrage condition under the efficient condition \textbf{\emph{(NAEF)}} is satisfied.
\end{theorem}

\begin{proof}
Necessity is shown in Proposition \ref{proposition:ConsistentRN}, so we only prove sufficiency.
Suppose that \textbf{NAEF} is satisfied.
According to Theorem~\ref{Theorem: FFTAP}, there exists a risk-neutral measure $\bQ$.
By Lemma~\ref{RNInequalitiesST},
  \begin{align*}
  & P^{bid, j, *}_{\sigma_1} \leq \mathbb{E}_{\bQ}\Big[P^{ask, j, *}_{\sigma_2} +\sum_{u=\sigma_1+1}^{\sigma_2}A^{j, *}_u\Big | \cF_{\sigma_1}\Big],  \quad P^{ask, j, *}_{\sigma_1} \geq \mathbb{E}_{\bQ}\Big[P^{bid, j, *}_{\sigma_2} +\sum_{u=\sigma_1+1}^{\sigma_2}  A^{j, *}_u\Big | \cF_{\sigma_1}\Big],
\end{align*}
 for all $j \in \cJ^*$ and stopping times $0 \leq \sigma_1 <\sigma_2 \leq T$.
Now, let us define the processes
\begin{align}
Y^{b, j}_t&:=\underset{\sigma \in \cT_t}{\esssup}\; \mathbb{E}_{\bQ}\Big[P^{bid, j, *}_{\sigma}+ \sum_{u=t+1}^{\sigma}A^{j, *}_u \Big | \cF_t\Big], \nonumber \\
Y^{a, j}_t&:=\underset{\sigma \in \cT_t}{\essinf}\; \mathbb{E}_{\bQ}\Big[P^{ask, j, *}_{\sigma}+ \sum_{u=t+1}^{\sigma}A^{j, *}_u \Big | \cF_t\Big],   \nonumber\\
X^{b, j}_t&:=Y^{b, j}_t+\sum_{u=1}^tA^{ j, *}_t, \quad X^{a, j}_t:=Y^{a, j}_t+\sum_{u=1}^tA^{ j, *}_t, \label{eq:defofY}
\end{align}
for all $t \in \cT$ and $j \in \cJ^*$.
From Lemma~\ref{lemma: Inequalities1}, we know that under $\bQ$ the process $X^a$ is a submartingale and the process $X^b$ is a supermartingale, and that they satisfy $X^b \leq X^a$.

For $t =0, 1, \dots, T-1$ and $j \in \cJ^*$, recursively define
\begin{align}\label{eq:consistent1}
M^j_0&:=Y^{a,j}_0, \quad P^j_0 :=Y^{a, j}_0, \quad P^j_{t+1}:=\lambda^j_t Y^{a, j}_{t+1}+(1-\lambda^j_t)Y^{b, j}_{t+1},\\ \notag
M^j_{t+1}&:=P^{j}_{t+1}+\sum_{u=1}^{t+1}A^{ j}_{u},
\end{align}
where $\lambda^j_t$ satisfies
\begin{equation}
\lambda^j_{t}=
\begin{dcases}\label{eq:defofLambda}
\frac{M^j_t-\mathbb{E}_{\bQ}[ X^{b, j}_{t+1}| \cF_t]}{\mathbb{E}_{\bQ}[X^{a, j}_{t+1}-X^{b, j}_{t+1}| \cF_t]},  &\textrm{if} \quad \mathbb{E}_{\bQ}[X^{a, j}_{t+1}| \cF_t] \neq \mathbb{E}_{\bQ}[ X^{b, j}_{t+1}| \cF_t], \\
 \frac12,  &\quad \textrm{otherwise}.
\end{dcases}
\end{equation}

\noindent
Lets fix $j\in\cJ^*$ for the rest of the proof.

\noindent
\textbf{Step 1:}
In this step, we show that the processes above are well defined and adapted.
First, note that $P_0$ and $M_0$ are well defined, and that, by \eqref{eq:defofLambda},
\begin{equation*}
  \lambda^j_0=\frac{M^j_0-\mathbb{E}_{\bQ}[ X^{b, j}_{1}| \cF_0]}{\mathbb{E}_{\bQ}[X^{a, j}_{1}- X^{b, j}_{1}| \cF_0]}\;, \quad \textrm{or} \quad \lambda^j_0=\frac12.
\end{equation*}
Thus, $\lambda^j_0$ is well defined and $\cF_0$-measurable.
Next, we compute $P_1^j$ and $M_1^j$, and consequently we compute $\lambda_1^j$; all of them being $\cF_1$-measurable.
Inductively,  we see that $P_t^j$, $M_t^j$, and  $\lambda_t^j$,  for $t=2,\ldots,T$ are well defined and $\cF_t$-measurable.

\noindent
\textbf{Step 2:}
We inductively show that $\lambda^j_t\in [0, 1]$ for $t=0, 1, \dots, T-1$.
We first show that $\lambda^j_0\in [0, 1]$.
If $\mathbb{E}_{\bQ}[X^{a, j}_{1}-X^{b, j}_{1}|\cF_0]=0$, then $\lambda^j_0\in [0, 1]$ automatically, so suppose that $\mathbb{E}_{\bQ}[X^{a, j}_{1}-X^{b, j}_{1}|\cF_0]>0$.
Now, by the definition of $M^j$, we have that $M^j_0=X^{a, j}_0$, so \eqref{eq:defofLambda} gives that
\begin{equation}\label{eq:Consistent1}
\lambda^j_0=\frac{X^{a, j}_0-\mathbb{E}_{\bQ}[ X^{b, j}_{1}| \cF_0]}{\mathbb{E}_{\bQ}[X^{a, j}_{1}- X^{b, j}_{1}| \cF_0]}\;.
\end{equation}
The process $X^{a, j}$ is a submartingale under $\bQ$, so it immediately follows that $\lambda^j_0 \leq 1$.
On the other hand, since $X^{b, j}$ is a supermartingale under $\bQ$,
\begin{equation*}
\lambda^j_0\geq \frac{X^{a, j}_0- X^{b, j}_{0}}{\mathbb{E}_{\bQ}[X^{a, j}_{1}- X^{b, j}_{1}| \cF_0]}\;.
\end{equation*}
Because $X^{a, j}_0 \geq X^{b, j}_0$, we deduce that $\lambda^j_0 \geq 0$.

Suppose that $\lambda^j_t \in [0, 1]$ for $t=0, 1, \dots, T-2$.
We now prove that $\lambda^j_{T-1} \in [0,1]$.
If $\bE_\bQ[X^{a, j}_T-X^{b, j}_T |\cF_{T-1}]=0$, then $\lambda^j_{T-1}=1/2$, so assume that $\bE_\bQ[X^{a, j}_T-X^{b, j}_T |\cF_{T-1}]>0$.
According to \eqref{eq:defofLambda} and the definition of $M^j$, we have that
\begin{equation}\label{eq:eqofLambda2}
  \lambda^j_{T-1}=\frac{\lambda^j_{T-2}X^{a, j}_{T-1}+(1-\lambda^j_{T-2})X^{b, j}_{T-1}-\bE_\bQ[X^{b, j}_T |\cF_{T-1}]}{\bE_\bQ[X^{a, j}_T- X^{b, j}_T |\cF_{T-1}]}\;.
\end{equation}
Since $\lambda^j_{T-2} \leq 1$, and because $X^{b, j}$ is a supermartingale under $\bQ$, we have that
\begin{equation*}
  \lambda^j_{T-1}\geq \frac{\lambda^j_{T-2}(X^{a, j}_{T-1}-X^{b, j}_{T-1})}{\bE_\bQ[X^{a, j}_T- X^{b, j}_T |\cF_{T-1}]}\;.
\end{equation*}
Because $X^{a, j} \geq X^{b, j}$, we arrive at $\lambda^j_{T-1} \geq 0$.
Now, since $X^{a,j}_{T-1} \geq X^{b,j}_{T-1}$ and $\lambda^j_{T-2} \leq 1$, we see from \eqref{eq:eqofLambda2} that
\begin{equation*}
  \lambda^j_{T-1}\leq \frac{X^{a, j}_{T-1}-\bE_\bQ[X^{b, j}_T |\cF_{T-1}]}{\bE_\bQ[X^{a, j}_T- X^{b, j}_T |\cF_{T-1}]}\;.
\end{equation*}
The process $X^{a,j}$ is a submartingale under $\bQ$, so it follows that $\lambda^j_{T-1} \leq 1$.
We conclude that $\lambda^j_t \in [0, 1]$ for $t=0, 1, \dots, T-1$.

\noindent
\textbf{Step 3:}
Next, we show that $M$ is a martingale under $\bQ$.
First we note that by \eqref{eq:defofY} and \eqref{eq:consistent1} we have
\begin{equation}\label{eq:MandX}
M_{t+1}^j = \lambda_t^j X_{t+1}^{a,j} + (1-\lambda_t^j) X_{t+1}^{b,j}.
\end{equation}
From here, the $\bQ$-integrability of $M^j$ follows from  $\bQ$-integrability of $X^{a,j}, X^{b,j}$ and boundedness of  $\lambda^j$.
From \eqref{eq:defofLambda} and \eqref{eq:MandX}, we get that $\bE_\bQ[M^j_{t+1} | \cF_t]=M^j_t$, for $t =0, 1, \dots, T-1$. Hence, $M^j$ is a martingale under $\bQ$.

\noindent
\textbf{Step 4:}
We continue by showing that $P^j$ satisfies $P^{bid, *, j} \leq P^j \leq P^{ask, *, j}$.
Let us first show that $P^{bid,j}_0 \leq P^j_0 \leq P^{ask, j}_0$.
By definition of $P^j_0$, we have that $P^j_0=Y^{a, j}_0$, and by \eqref{eq:defofY} we see that $Y^{a, j}_0=X^{a,j}_0$.
Therefore, the claim holds since $P^{bid, j}_0 \leq  X^{a, j}_0 \leq P^{ask, j}_0$.

We proceed by proving that $P^{bid,j}_t \leq P^j_t \leq P^{ask, j}_t$ for all $t \in \{1, \dots, T\}$.
Towards this, let $t \in \{1, \dots, T\}$.
By the definition of $P^j_{t}$, we have $P^j_{t}=\lambda^j_{t-1}Y^{a, j}_{t}+(1-\lambda^j_{t-1})Y^{b, j}_{t}$.
From \eqref{eq:defofY}, it is true that $X^{a, j}_t \geq X^{b,j}_t$ if and only if $Y^{a, j}_t \geq Y^{b, j}_t$.
Also, since $t \in \cT_t$, we see from \eqref{eq:defofY}  that $Y^{b, j}_t\geq  P^{bid, j, *}_t$ and $Y^{a, j}_t \leq P^{ask, j}_t$.
According to Step 1,  $\lambda^j_{t-1} \in [0,1]$.
So, putting everything together, we obtain
\begin{equation*}
  P^{bid, j}_t \leq Y^{b, j}_t \leq P^j_t \leq Y^{a, j}_t \leq P^{ask, j}_t.
\end{equation*}
We conclude that $\{\bQ, P, A, M\}$  is a CPS.
\end{proof}

\section{Superhedging and subhedging theorem}\label{section:SuperHedging}

In this section, we  define the superhedging ask and subhedging bid prices for a dividend-paying contingent claim, and then we provide  an important representation theorem for these prices.
The representation theorem is important because it provides an alternative way of computing the superhedging ask and superhedging bid prices.
 Also, it is an application of the Fundamental Theorem of Asset Pricing: the theorem relates how the no-arbitrage condition (and hence the existence of risk-neutral measures) is related to the pricing of contingent claims.

 For results related to this topic,  both for discrete-time and continuous-time markets with transaction costs, we refer to, among others,  Soner, Shreve, and Cvitanic~\cite{SonerShreveCvitanic1995}; Levental and Skorohod~\cite{LeventalSkorohod1997}; Cvitanic, Pham, and Touzi~\cite{CvitanicPhamTouzi1999}; Touzi~\cite{Touzi1999}; Bouchard and Touzi~\cite{BouchardTouzi2000}; Kabanov, R\'{a}sonyi, and Stricker~\cite{KabanovRasonyiStricker2002};  Schachermayer~\cite{Schachermayer2004}; Campi and Schachermayer~\cite{CampiSchachermayer2006};  Cherny~\cite{Cherny2007a}; Pennanen~\cite{Pennanen2011, Pennanen2011b,Pennanen2011c,Pennanen2011d}.
Our contribution to this literature is that we consider dividend-paying securities such as swap contracts  as hedging securities.

A \emph{contingent claim} $D$ is any a.s. bounded, $\bR$-valued, $\bF$-adapted process.
Here, $D$ is interpreted as the spot cash flow process (not the cumulative cash flow process).
We remark that the boundedness assumption on contingent claims is satisfied for fixed income securities.

Let us now define the set of self-financing trading strategies initiated at time $t \in \{0,1, \dots, T-1\}$ with bounded components ($j=1, \dots, N$) as
\begin{equation*}
  \cS(t):=\Big\{ \phi \in \cS  : \phi^j \text{ is bounded a.s. for } j \in \cJ^*, \; \phi_s=0 \textrm{ for all } s \leq t\Big\},
\end{equation*}
and the set of attainable claims at zero cost initiated at time $t \in \{0,1, \dots, T-1\}$ as
\begin{align*}
  \cK(t):=\Big\{V^*_T(\phi): \phi \in \cS(t)  \; \text{such that } V_0(\phi)=0\Big\}.
\end{align*}

\begin{remark}\label{remark:S(t)}\mbox{}
\begin{itemize}
\item[(i)]
$\cS(t)$ and $\cK(t)$ are closed with respect to multiplication by  random variables in\\ $L^{\infty}_+(\Omega, \cF_t, \bP; \bR)$.\footnote{$L^{\infty}_+(\Omega, \cF_t, \bP; \bR):=\{X \in L^{\infty}(\Omega, \cF_t, \bP; \bR): X \geq 0\}.$}
\item[(ii)]
$\cS \supset \cS(0)\supset \cS(1)\supset \cdots \supset \cS(T-1)$ and $ \cK \supset \cK(0) \supset \cK(1)  \supset \cdots \supset\cK(T-1)$.
Moreover, if $\bQ$ is a risk-neutral measure, then $\bE_\bQ[K] \leq 0$ for all $K \in \cK(t)$, for $t=0, 1, \dots, T-1$.
\end{itemize}
\end{remark}

\noindent
We proceed by defining the main objects of this section.

\begin{definition}\label{def:SuperHedgingPrices}
  The discounted superhedging ask and subhedging bid prices of a contingent claim $D$ at time $t \in \{0, \dots, T-1\}$ are defined as $\pi^{ask}_t(D):=\essinf \: \cW^a(t, D)$ and $\pi^{bid}_t(D):=\esssup \: \cW^b(t, D)$, where
  \begin{align*}
    &\cW^a(t, D):= \Big\{W \in L^0(\Omega, \cF_t, \bP; \bR) : -W+\sum_{u=t+1}^TD^*_u  \in \cK(t) -L^0_+(\Omega, \cF_T, \bP; \bR)\Big\},\\
    &\cW^b(t,D):=\Big\{W \in L^0(\Omega, \cF_t, \bP; \bR) : W-\sum_{u=t+1}^TD^*_u  \in \cK(t) -L^0_+(\Omega, \cF_T, \bP; \bR)\Big\}.
  \end{align*}
\end{definition}

\noindent
Note that $\pi^{ask}_t(D)=-\pi^{bid}_t(-D)$ and
  \begin{align*}
    &\cW^a(t, D)= \Big\{W \in L^0(\Omega, \cF_t, \bP; \bR) : \: \exists \: K \in \cK(t) \text{ such that }\sum_{u=t+1}^TD^*_u  \leq K+W\Big\},\\
    &\cW^b(t,D)=\Big\{W \in L^0(\Omega, \cF_t, \bP; \bR) : \: \exists \: K \in \cK(t) \text{ such that } -\sum_{u=t+1}^TD^*_u  \leq K-W \Big\}.
  \end{align*}

\noindent
\begin{remark}\label{remark:SuperNAEF}\mbox{}
  \begin{itemize}
  \item[(i)]
 For each $t \in \{0, 1, \dots, T-1\}$, the prices $\pi^{ask}_t(D)$ and $\pi^{bid}_t(D)$ have the following interpretations:
   The price $\pi^{ask}_t(D)$ is interpreted as the least discounted cash amount  $W$ at time $t$ so that the gain $-W+\sum_{u=t+1}^TD^*_u$ can be superhedged at zero cost.
  On the other hand,  the random variable $\pi^{bid}_t(D)$ is interpreted as the greatest discounted cash amount $W$ at time $t$  so that the gain $W-\sum_{u=t+1}^TD^*_u$ can be superhedged at zero cost.
    \item[(ii)]
In view of (i) above,  it is unreasonable for the discounted ex-dividend ask price at time $t\in \{0, 1, \dots, T-1\}$  of a contingent claim $D$ to be a.s. greater than $\pi^{ask}_t(D)$, and for the ex-dividend bid price at time $t\in \{0, 1, \dots, T-1\}$  of a contingent claim $D$ to be a.s. less than $\pi^{bid}_t(D)$.
\item[(iii)]
Direction of trade matters: a market participant can buy a contingent claim $D$ at price $\pi^{ask}_t(D)$ and  sell $D$ at price $\pi^{bid}_t(D)$.
This is in contrast to frictionless markets, where a contingent claim can be bought and sold at the same price.
\item[(iv)]
The prices $\pi^{ask}_t(D)$ and $\pi^{bid}_t(D)$ satisfy $\pi^{ask}_t(D)<\infty$ and $\pi^{bid}_t(D)>-\infty$.
 Indeed, since $0 \in \cK(t)$, $1 \in L^0_{++}(\Omega, \cF_t, \bP; \bR)$,  and $\sum_{u=t+1}^TD^*_u$ is a.s. bounded, say by $M$, we have that $-M+\sum_{u=t+1}^TD^*_u \in L^0_-(\Omega, \cF_T, \bP; \bR)$.
 Thus, $\pi^{ask}_t(D)\leq M$.
 Similarly, $\pi^{bid}_t(D) \geq -M$.
  \end{itemize}
\end{remark}

\noindent Next, we define \emph{the sets of extended attainable claims initiated at time $t \in \{0, 1, \dots, T-1\}$ associated with cash amount $W \in L^0(\Omega, \cF_t, \bP; \bR)$}:
\begin{align*}
  \cK^a(t, W)&:=\cK(t)+\Big\{\xi \Big( -W+\sum_{u=t+1}^T D^*_u\Big):\;  \xi \in L^{\infty}_+(\Omega, \cF_t, \bP; \bR) \Big\}, \\
  \cK^b(t, W)&:= \cK(t)+\Big\{\xi\Big( W-\sum_{u=t+1}^T D^*_u\Big):\; \xi \in L^{\infty}_+(\Omega, \cF_t, \bP; \bR)\Big\}.
\end{align*}

\begin{remark}\mbox{}
\begin{itemize}
\item[(i)]
The sets $\cK^a(t, W)$ and $\cK^b(t, W)$ are closed with respect to multiplication by random variables in the set $L^{\infty}_+(\Omega, \cF_t, \bP; \bR)$, and in view of Lemma~\ref{Arbitrage/LemmaConvexCone} they are convex cones.
Also, $\cK(t) \subset \cK^a(t, W)\cap \cK^b(t, W)$ since $0\in L^{\infty}_+(\Omega, \cF_t, \bP; \bR)$.
\item[(ii)]
In view of Proposition~\ref{Intro/LemmaDiscountedSelfFinancing},
\begin{equation}\label{eq:Super1}
  \Big\{\xi \Big( -\pi^{ask}_t(D)+\sum_{u=t+1}^T D^*_u\Big):\;  \xi \in L^{\infty}_+(\Omega, \cF_t, \bP; \bR) \Big\}
\end{equation}
is the set of all discounted terminal values associated with zero-cost, self-financing, buy-and-hold trading strategies in the contingent claim $D$ with  discounted ex-dividend ask price $\pi^{ask}_t(D)$.
On the other hand, the convex cone
\begin{equation}\label{eq:Super2}
  \Big\{\xi \Big( \pi^{bid}_t(D)-\sum_{u=t+1}^T D^*_u\Big):\;  \xi \in L^{\infty}_+(\Omega, \cF_t, \bP; \bR) \Big\}
\end{equation}
is the set of all discounted terminal values associated with zero-cost, self-financing, sell-and-hold trading strategies in the contingent claim $D$ with discounted ex-dividend bid price $\pi^{bid}_t(D)$.
\end{itemize}
\end{remark}

We will now introduce definitions related to the sets of extended attainable claims.
For each $t \in \{0, 1, \dots, T-1\}$ and $X \in L^0(\Omega, \cF_t, \bP; \bR)$, a probability measure $\bQ$ is \emph{risk-neutral for} $\cK^a(t, X)$ ($\cK^b(t, X)$) if $\bQ \in \cZ$ and $X$ is $\bQ$-integrable, and if $\bE_\bQ[K] \leq 0$ for all $K \in \cK^a(t, X)$ ($K \in \cK^b(t, X)$).
We denote by $\cR^a(t, X)$ ($\cR^b(t, X)$) the set of all risk-neutral measures $\bQ$ for $\cK^a(t, X)$ ($\cK^b(t, X)$) so that $\d \bQ /\d \bP \in L^{\infty}(\Omega, \cF_T, \bP; \bR)$.
We say that \textbf{NA} holds for $\cK^a(t, X)$ if $\cK^a(t, X) \cap L^0_+(\Omega, \cF_T, \bP; \bR)=\{0\}$, and likewise we say that \textbf{NA} holds for $\cK^b(t, X)$ if $\cK^b(t, X) \cap L^0_+(\Omega, \cF_T, \bP; \bR)=\{0\}$.

\noindent
We will say that $\cK^a(t, X)$  satisfies \textbf{EF} if
\begin{equation*}
  \Big\{(\phi, \xi) \in \cS(t) \times L^{\infty}_+(\Omega, \cF_t, \bP; \bR):V_0(\phi)=0, \; V^*_T(\phi)+\xi \Big(-X+\sum_{u=t+1}^T D^*_u\Big)=0\Big\}=\{(0, 0)\},
\end{equation*}
and say that $\cK^b(t, X)$ satisfies \textbf{EF} if
 \begin{equation*}
\Big  \{(\phi, \xi) \in \cS(t) \times L^{\infty}_+(\Omega, \cF_t, \bP; \bR): V_0(\phi)=0, \; V^*_T(\phi)+\xi \Big( X- \sum_{u=t+1}^T D^*_u\Big)=0\Big\}=\{(0, 0)\}.
\end{equation*}

\begin{remark}
According to Lemma~\ref{lemma:ClosednessExtended}, for any $t \in \{0, 1, \dots, T-1\}$ and  $X \in L^0(\Omega, \cF_t, \bP; \bP)$,  \textbf{NAEF} holds for $\cK^a(t, X)$ ($\cK^b(t, X)$) if and only if $\cR^a(t, X) \neq \emptyset$ ($\cR^b(t, X) \neq \emptyset$).
\end{remark}

For each $t \in \{0, 1, \dots, T-1\}$, we denote by $\cR(t)$ the set of all risk-neutral measures for $\cK(t)$ so that $\d \bQ /\d \bP \in L^{\infty}(\Omega, \cF_T, \bP; \bR)$.
Specifically, we define $\cR(t)$ as
\begin{align*}
  \cR(t)&:=\Big\{\bQ \in \cZ:\; \bE_\bQ[K] \leq 0 \text{ for all $K \in \cK(t)$}\Big\} .
\end{align*}
We note that $\cR^a(t, X)\cup \cR^b(t, X) \subseteq \cR(t)$ for any $X \in  L^0(\Omega, \cF_t, \bP; \bR)$ since $\cK(t) \subseteq \cK^a(t, X)\cap \cK^b(t, X)$.
Also, by the definition of a risk-neutral measure, it immediately follows that any risk-neutral measure $\bQ$ (as in Definition~\ref{def:RiskNeutral}) satisfies $\bQ\in \cR(t)$ for any $t \in \{0,1, \dots, T-1\}$.

The next technical lemma is needed to derive the dual representations of the superhedging ask and subhedging bid prices.

\begin{lemma}\label{lemma:RNnontrivial}\mbox{}
\begin{itemize}
\item[(i)]
For each $t \in \{0, 1, \dots, T-1\}$, if $\cR(t) \neq \emptyset$ and $\bQ \in \cR(t)$, then we have that $\bE_\bQ[K|\cF_t] \leq 0$ $\bQ$-a.s. for all $K \in \cK(t)$.
\item[(ii)]
For each $t \in \{0, 1, \dots, T-1\}$ and  $X \in L^0(\Omega, \cF_t, \bP; \bR)$, if $\cR^a(t, X) \neq \emptyset$ and $\bQ \in \cR^a(t, X)$, then  we have that $\bE_{\bQ}[K^a|\cF_t] \leq 0$ $\bQ$-a.s. for all $K^a \in \cK^a(t, X)$.
\item[(iii)]
For each $t \in \{0, 1, \dots, T-1\}$ and  $X \in L^0(\Omega, \cF_t, \bP; \bR)$, if $\cR^b(t, X) \neq \emptyset$ and $\bQ \in \cR^b(t, X)$, then  we have that $\bE_{\bQ}[K^b|\cF_t] \leq 0$ $\bQ$-a.s. for all $K^b \in \cK^b(t, X)$.
\end{itemize}
\end{lemma}

\begin{proof}
We only prove \emph{(i)} and \emph{(ii)}.
The proof of \emph{(iii)} is very similar to the proof of \emph{(ii)}.
We fix $t \in \{1, \dots, T-1\}$ throughout the proof.
Observe that in view of Lemma~\ref{lemma:RNIntegrability}, we have that for each $\bQ \in \cR(t)$,  any $K \in \cK$ is $\bQ$-integrable.
 Moreover, because any contingent claim is bounded a.s.,  for each $\bQ \in \cR^a(t, X)$ ($\bQ \in \cR^b(t, X)$), any $K^a \in \cK^a(t, X)$ ($K^b \in \cK^b(t, W)$) is $\bQ$-integrable.

 \noindent
\emph{(i)}:
We prove by contradiction.
Let  $\bQ \in \cR(t)$, and suppose that there exists   and $K \in \cK(t)$ such that $\bE_{\bQ}[K| \cF_t](\omega)>0$ for all $ \omega \in \Omega^t$, where $\Omega^t \subseteq \Omega$ and $\bP(\Omega^t)>0$.
Note that $\Omega^t \in \cF_t$ since $\bE_{\bQ}[K|\cF_t]$ is $\cF_t$-measurable.
By definition of $\cK(t)$, there exists $\phi \in \cS(t)$ with $V_0(\phi)=0$ such that $K=V^*_T(\phi)$.
Define the process $\psi:=1_{\Omega^t}\phi$.
 Since $\Omega^t$ is $\cF_{t}$-measurable and $\cS(t)$ is closed with respect to multiplication by random variables in the set $L^{\infty}_+(\Omega, \cF_t, \bP; \bR)$, we have that $\psi \in \cS(t)$.
  Moreover, $V_0(\psi)=1_{\Omega^t}V_0(\phi)=0$ because $1_{\Omega^t}$ is nonnegative.
  Therefore, $V^*_T(\psi) \in \cK(t)$.
Since $V^*_T(\psi)=1_{\Omega^t}V^*_T(\phi)=1_{\Omega^t}K$, we have that $\bE_\bQ[V^*_T(\psi)]=\bE_{\bQ}[1_{\Omega^t} \bE_\bQ[K|\cF_t]]>0$, which contradicts that $\bQ \in \cR(t)$.

\noindent
\emph{(ii)}:
As in \emph{(i)}, we will prove by contradiction.
Let  $X \in L^0(\Omega, \cF_t, \bP; \bR)$ and  $\bQ \in \cR^a(t, X)$, and assume that there exist  $K \in \cK(t)$ and $\xi \in  L^{\infty}_+(\Omega, \cF_t, \bP; \bR)$ such that
\begin{equation*}
-\xi(\omega) X(\omega)+\bE_\bQ\Big[K+\xi \sum_{u=t+1}^T D^*_u\Big|\cF_t\Big](\omega) > 0, \quad  \omega \in \Omega^t,
\end{equation*}
where $\Omega^t \subseteq \Omega$ and $\bP(\Omega^t)>0$.
 Since  $\cR^a(t,X) \subset \cR(t)$, we have that $\bQ \in \cR(t)$.
 In view of \emph{(i)} above, it follows that $\bE_\bQ[K|\cF_t] \leq 0$.
 Thus,
 \begin{equation*}
-\xi(\omega) X(\omega)+\xi(\omega)\bE_\bQ\Big[ \sum_{u=t+1}^T D^*_u\Big|\cF_t\Big](\omega) > 0, \quad  \omega \in \Omega^t.
\end{equation*}
We proceed by defining the $\cF_t$-measurable random variable $\vartheta:=1_{\Omega^t}\xi$.
Because $\Omega^t \in \cF_t$, it is true that $\vartheta \in L^{\infty}_+(\Omega, \cF_t, \bP; \bR)$.
Now, by the tower property of conditional expectations we obtain
\begin{equation*}
  \bE_{\bQ} \Big[\vartheta \Big( -X+\sum_{u=t+1}^T D^*_u\Big)\Big]=  \bE_{\bQ} \Bigg[1_{\Omega^t}\Big( -\xi X+\xi\bE_{\bQ} \Big[\sum_{u=t+1}^T D^*_u\Big| \cF_t\Big]\Big)\Bigg]>0.
\end{equation*}
This contradicts that $\bQ \in \cR^a(t, X )$ since $\vartheta \in L^{\infty}_+(\Omega, \cF_t, \bP; \bR)$ and $0 \in \cK(t)$.
\end{proof}

We are ready to prove the main result of this section: the dual representations of the superhedging ask price and subhedging bid price.

\begin{theorem}\label{theorem:SuperhedgingNoNTrivial}
Suppose that the no-arbitrage condition under the efficient friction assumption \textbf{\emph{(NAEF)}} is satisfied.
Let $t \in \{0, 1, \dots, T-1\}$ and $D$ be a contingent claim.
Then the following hold:
\begin{itemize}
 \item[(i)]
The essential infimum of $\cW^a(t, D)$ and the essential supremum of $\cW^b(t, D)$ are attained.
\item[(ii)]
Suppose that for each $t \in \{0, 1, \dots, T-1\}$ and $X \in L^0(\Omega, \cF_t, \bP; \bR)$,  the efficient friction assumption \textbf{\emph{(EF)}} holds for $\cK^a(t, X)$ and $\cK^b(t, X)$.
Then the discounted superhedging ask  and subhedging bid prices for contingent claim $D$ at time $t$ satisfy
 \begin{align}
    &\pi^{ask}_t(D)= \;\underset{\bQ \in \cR(t)}{\emph{\esssup}}\; \bE_\bQ \Big[\sum_{u=t+1}^T D^*_u \Big | \cF_t\Big],\label{eq:SuperhedgingAsk1NT}\\
    &\pi^{bid}_t(D)=  \; \underset{\bQ \in \cR(t)}{\emph{\essinf}} \;\bE_\bQ \Big[\sum_{u=t+1}^T D^*_u \Big | \cF_t\Big].\label{eq:SuperhedgingBid1NT}
 \end{align}
 \end{itemize}
\end{theorem}

\begin{proof}
Since $\pi^{ask}_t(D)=-\pi^{bid}_t(-D)$ holds for all $t \in \{0, \dots, T-1\}$ and contingent claim $D$, it suffices to show that the essential infimum of $\cW^a(t, D)$  is attained and \eqref{eq:SuperhedgingAsk1NT} holds.
Let us fix $t \in \{0, 1, \dots, T-1\}$ throughout the proof.

\noindent
We first prove \emph{(i)}.
 Let $W^m$ be a sequence decreasing a.s. to $\pi^{ask}_t(D)$, and fix $K^m \in \cK(t)$ and $Z^m \in L^0_+(\Omega, \cF_T, \bP; \bR)$ so that $-W^m+\sum_{u=t+1}^T D^*_u=K^m-Z^m$.
Since a.s. converges implies convergence in probability, we see that the sequence $K^m-Z^m$ converges in probability to some $Y$.
Due to Theorem~\ref{lemma:KClosed}, we have that $\cK(t) -L^0_+(\Omega, \cF_T, \bP; \bR)$ is $\bP$-closed.
Therefore, $Y \in \cK(t) -L^0_+(\Omega, \cF_T, \bP; \bR)$.
This proves that $-\pi^{ask}_t(D)+\sum_{u=t+1}^T D^*_u \in \cK(t)-L^0_+(\Omega, \cF_T, \bP; \bR)$.

\noindent
Next, we show that \emph{(ii)} holds.
We begin by showing that
\begin{equation} \label{eq:Super11}
  \pi^{ask}_t(D)\geq \underset{\bQ \in \cR(t)}{\esssup}\ \bE_\bQ \Big[\sum_{u=t+1}^T D^*_u \Big | \cF_t\Big].
\end{equation}
By \emph{(i)}, we have that $\pi^{ask}_t(D) \in \cW^a(t, D)$, so there exists   $K^* \in \cK(t)$ so that
\begin{equation}\label{eq:Super9}
K^*+\pi^{ask}_t(D)-\sum_{u=t+1}^T D^*_u\geq 0.
\end{equation}
We are assuming that \textbf{NAEF} is satisfied, so according to Theorem~\ref{Theorem: FFTAP} there exists a risk-neutral measure $\bQ^*$.
Because any risk-neutral measure $\bQ$ satisfies $\bQ \in \cR(t)$, we obtain that $\bQ^* \in \cR(t)$.
By taking the conditional expectation with respect to $\cF_t$ under $\bQ^*$ of both sides of the last inequality we deduce that
\begin{equation*}
\pi^{ask}_t(D)+\bE_{\bQ^*} [K^*| \cF_t]\geq \bE_{\bQ^*} \Big[\sum_{u=t+1}^T D^*_u\Big | \cF_t \Big].
\end{equation*}
According to part \emph{(i)} of Lemma~\ref{lemma:RNnontrivial}, we have that  $\bE_{\bQ^*}[K^*|\cF_t] \leq 0$.
As a result,  $\pi^{ask}_t(D) \geq\bE_{\bQ^*} \big[\sum_{u=t+1}^T D^*_u \big| \cF_t\big]$.
Taking the essential supremum of both sides of the last inequality over $\cR(t)$ proves that \eqref{eq:Super11} holds.

Next, we show that
\begin{equation}\label{eq:Super12}
  \pi^{ask}_t(D)\leq \underset{\bQ \in \cR(t)}{\esssup}\: \bE_\bQ \Big[\sum_{u=t+1}^T D^*_u\Big| \cF_t \Big].
\end{equation}
By \emph{(i)}, we have that $\pi^{ask}_t(D) > -\infty$, so we may take  $X \in L^0(\Omega, \cF_t, \bP; \bR)$ so that $\pi^{ask}_t(D) > X$.
We now prove by contradiction that \textbf{NA} holds for $\cK^b(t, X)$.
Towards this aim, we assume that there exist $K \in \cK(t)$, $\xi \in L^{\infty}_+(\Omega, \cF_t, \bP; \bR)$, and $\Omega^0 \subseteq \Omega$ with $\bP(\Omega^0)>0$ so that
\begin{equation}\label{eq:Super8}
  K+\xi\Big(X-\sum_{u=t+1}^TD^*_u\Big) \geq 0 \quad \text{ a.s.}, \quad  K+\xi\Big(X-\sum_{u=t+1}^TD^*_u\Big) > 0 \quad \text{a.s. on } \Omega^0.
\end{equation}
Since \textbf{NA} is satisfied for underlying market $\cK$, we have from \eqref{eq:Super8}  that there exists $\Omega^1 \subseteq \Omega^0$ with $\bP(\Omega^1)>0$ such that $\Omega^1 \in \cF_t$ and $\xi >0$ a.s. on $\Omega^1$.
 Otherwise, our assumption that \textbf{NA} holds is contradicted.
  Of course, if $\Omega^1 \subseteq \Omega^0$ is any set such that $\Omega^1 \in \cF_t$,  $\bP(\Omega^1)>0$, and $\xi=0$ a.s. on $\Omega^1$, then $1_{\Omega^1}K \in \cK(t) \in \cK$ satisfies $1_{\Omega^1}K \geq 0$ a.s., and $1_{\Omega^1}K>0$ a.s. on $\Omega^1$, which violates that \textbf{NA} is satisfied.

Moreover, we observe that $X-\sum_{u=t+1}^TD^*_u \geq 0$ a.s. on $\Omega^1$.
If there exists $\Omega^2 \subseteq \Omega^1$   with $\bP(\Omega^2)>0$ such that $\Omega^2 \in \cF_t$ and  $X-\sum_{u=t+1}^TD^*_u < 0$ a.s. on $\Omega^2$, then
from \eqref{eq:Super8} we see that $K  \geq $ a.s., and $K>0$ a.s. on $\Omega^2$, which contradicts that \textbf{NA} holds for $\cK$.

Now, let us define
\begin{equation*}
  \widetilde{X}:=1_{\Omega^1}X+1_{(\Omega^{1})^c}\pi^{ask}_t(D), \qquad \widetilde{K}:=1_{\Omega^1}\frac{K}{\sup_{\omega \in \Omega^1}\{\xi(\omega)\}}+1_{(\Omega^1)^c}K^*.
\end{equation*}
From \eqref{eq:Super9} we immediately have that
\begin{align*}
&\widetilde{K}+\widetilde{X}-\sum_{u=t+1}^T D^*_u=K^*+\pi^{ask}_t(D)-\sum_{u=t+1}^TD^*_u \geq 0& \quad &\text{a.s. on $(\Omega^1)^c$}.
\end{align*}
On the other hand, from \eqref{eq:Super8} and since  $X-\sum_{u=t+1}^TD^*_u \geq 0$ a.s. on $\Omega^1$, we see that
\begin{align*}
 &\widetilde{K}+\widetilde X-\sum_{u=t+1}^TD^*_u = \frac{K}{\sup_{\omega \in \Omega^1}\{\xi(\omega)\}}+X-\sum_{u=t+1}^T D^*_u  \geq 0& \quad &\text{a.s. on $\Omega^1$}.
\end{align*}
Consequently, $\widetilde{K}+\widetilde{X}-\sum_{u=t+1}^T D^*_u \geq0$ a.s. on $\Omega$.
Now, since $0 \leq 1/\sup_{\omega \in \Omega^1}\{\xi(\omega)\}<\infty$,  and because $\cK(t)$ is a convex cone that is closed with respect to multiplication by random variables in $L^{\infty}_+(\Omega, \cF_t, \bP; \bR)$, we have that $\widetilde{K} \in \cK(t)$.
Therefore  $\widetilde X \in \cW^a(t, D)$.
 However, since $\widetilde{X}$ satisfies $\widetilde X \leq \pi^{ask}_t(D)$ and $\bP(\widetilde{X}<\pi^{ask}_t(D))>0$, we have that $\widetilde X \in \cW^a(t, D)$ contradicts $\pi^{ask}_t(D)=\essinf\: \cW^a(t, D)$.
Thus, \textbf{NA} holds for $\cK^b(t, X)$.

By assumption, \textbf{EF} holds for $\cK^b(t, X)$, so \textbf{NAEF} is satisfied for $\cK^b(t, X)$.
According to Lemma~\ref{lemma:ClosednessExtended}  there exists  $\hat \bQ\in \cR^b(t, X)$.
In view the claim \emph{(iii)} in Lemma~\ref{lemma:RNnontrivial}, we see that
\begin{equation*}
\zeta X+\bE_{\hat \bQ}[K| \cF_t] \leq \zeta \: \bE_{\hat \bQ} \Big[\sum_{u=t+1}^TD^*_u\Big| \cF_t\Big], \quad  K \in \cK(t), \; \zeta \in L^{\infty}_+(\Omega, \cF_t, \bP; \bR).
\end{equation*}
Since $0 \in \cK(t)$ and $1 \in L^{\infty}_+(\Omega, \cF_t, \bP; \bR)$, we obtain that $ X\leq  \bE_{\hat \bQ} \big[\sum_{u=t+1}^TD^*_u\big| \cF_t\big]$.
Now, because $\cR^b(t, X) \subseteq \cR(t)$, we have that $\hat \bQ \in \cR(t)$.
Hence,
\begin{equation}\label{eq:Super5}
 X\leq \bE_{\hat \bQ} \Big[\sum_{u=t+1}^TD^*_u\Big| \cF_t\Big] \leq \underset{\bQ \in \cR(t)}{\sup}\: \bE_{\bQ} \Big[\sum_{u=t+1}^TD^*_u\Big| \cF_t\Big].
\end{equation}
The random variable $X<\pi^{ask}_t(D)$ is arbitrary, so for any scalar $\epsilon >0$ we may take $X:=\pi^{ask}_t(D)-\epsilon$.
From \eqref{eq:Super5}, we see that
\begin{equation*}
\pi^{ask}_t(D)\leq  \underset{\bQ \in \cR(t)}{\sup}\: \bE_{\bQ} \Big[\sum_{u=t+1}^TD^*_u\Big| \cF_t\Big]+\epsilon, \quad \epsilon>0.
\end{equation*}
Letting $\epsilon$ approach zero shows that \eqref{eq:Super12} holds.
This completes the proof of \emph{(i)}.
\end{proof}

\begin{remark}
An open question that remains is whether $\cR(t)$ can be replaced by $\cR(0)$ in the representations in Theorem~\ref{theorem:SuperhedgingNoNTrivial}.
In the arguments presented in this paper, it is more convenient to work with $\cR(t)$ than $\cR(0)$  because $\cK(t)$ is closed under multiplication by random variables in $L^{\infty}_+(\Omega, \cF_t, \bP; \bR)$.
In contrast, the set $\cK(0)$ is only closed under multiplication by random variables in $L^{\infty}_+(\Omega, \cF_0, \bP; \bR)$.
\end{remark}

\section{Conclusions}
In this paper, no-arbitrage pricing theory is extended to dividend-paying securities in discrete-time markets with transaction costs.
A version of the Fundamental Theorem of Asset Pricing is proved under the efficient friction assumption, and the representations for the superhedging ask and subhedging bid prices are given.
As usual, the proof of the Fundamental Theorem of Asset Pricing relies on showing that the set of all claims that can be superhedged at zero cost is closed under convergence in probability.
In the special case when there are no transaction costs on the dividends paid by the security, the no-arbitrage condition under the efficient friction assumption is proved to be equivalent to the existence of a consistent pricing system.
The general case, in which there are transaction costs on the dividends, is open.
The theory is motivated by credit default swaps and interest rate swaps.

\appendix
\section{Appendix}

\begin{lemma}\label{lemma:ClosednessExtended}
For each $t \in \{0, 1, \dots, T-1\}$ and $W \in L^0(\Omega, \cF_t, \bP; \bR)$, if the no-arbitrage condition under the efficient friction assumption is satisfied  for $\cK^a(t, W)$ $(\cK^b(t, W))$, then $\cR^a(t, W) \neq \emptyset$ $(\cR^b(t, W) \neq \emptyset)$.
\end{lemma}

\begin{proof}
Let us first fix $t \in \{0, 1, \dots, T-1\}$ and $W \in L^0(\Omega, \cF_t, \bP; \bR)$.
We only prove the lemma for $\cK^a(t, W)$, because the proof for $\cK^b(t, W)$ is similar.
Instead of working with $\cK^a(t, X)$, we will work with the more mathematically convenient set
\begin{equation*}
  \bK^a(t, W):=\big\{G(\phi, \xi, t, W) : \phi \in \cP(t), \xi \in L^{\infty}_+(\Omega, \cF_t, \bP; \bR)\big\},
\end{equation*}
where $\cP(t)$ is the set
\begin{equation*}
\cP(t):=\{\phi \in \cP: \phi^j \text{ is a.s. bounded for $j \in \cJ^*$, } \phi_s=1_{\{t+1 \leq s\}} \phi_s \text{ for all } s \in \cT^*\},
\end{equation*}
and
  \begin{align}
  G(\phi, \xi, t, W)&:=\sum_{j=1}^N  \phi^{j}_T\big(1_{\{ \phi^{j}_T \geq 0\}}P^{bid, j, *}_T+ 1_{\{ \phi^{j}_T < 0\}}P^{ask, j, *}_T\big)\notag\\
& \  -\sum_{j=1}^N\sum_{u=t+1}^T \Delta \phi^{ j}_u\big(1_{\{\Delta \phi^{j}_u \geq 0\}}P^{ask, j, *}_{u-1}+1_{\{\Delta \phi^{ j}_u< 0\}} P^{bid, j, *}_{u-1}\big)\notag\\
& \ +\sum_{j=1}^N \sum_{u=t+1}^T \phi^{ j}_u \big(1_{\{ \phi^{j}_u \geq 0\}} A^{ask,j, *}_u+1_{\{ \phi^{ j}_u < 0\}}A^{bid,j, *}_u\big)+\xi\Big( -W+\sum_{u=t+1}^T D^*_u \Big)
\label{eq:DefMappingF1}
\end{align}
for all for all $\bR^N$-valued stochastic processes
\begin{equation*}
 (\phi_s)_{s=1}^T \in L^0(\Omega, \cF_T, \bP; \bR^N) \times \cdots \times L^0(\Omega, \cF_T, \bP; \bR^N),
 \end{equation*}  and random variables $\xi\in L^{\infty}_+(\Omega, \cF_T, \bP; \bR)$.

Since $\bK^a(t, W)=\cK^a(t, W)$, we may equivalently prove that $\bK^a(t, W)-L^0_+(\Omega, \cF_T, \bP; \bR)$ is $\bP$-closed whenever \textbf{NAEF} is satisfied for $\cK^a(t, W)$.
Let $X^m \in  \bK^a(t, W)-L^0_+(\Omega, \cF_T, \bP; \bR)$ be a sequence converging in probability to some $X$.
We may find a subsequence $X^{k_m}$ that converges a.s. to $X$.
With an abuse of notation we denote this subsequence by $X^m$.
By the definition of $\bK^a(t, W)$, we may find $\phi^m \in \cP(t)$, $\xi^m\in L^0_+(\Omega, \cF_t, \bP; \bR)$, and $Z^m \in L^0_+(\Omega, \cF_T, \bP; \bR)$ so that $X^m=G(\phi^m, \xi^m, t, W)-Z^m$.
Using the same arguments as in Step 1 in the proof of Theorem~\ref{lemma:KClosed}, we prove that $\limsup_m \|\phi^m_s\|<\infty$ for all $t \in \cT^*$ and $\limsup_m \xi^m<\infty$.
  Then, we apply Lemma~\ref{Arbitrage/MeasurableSelectionShachProcess} to show that we may find a strictly increasing set of positive, integer-valued, $\cF_{T-1}$ measurable random variables $\sigma^m$ such that $\phi^{\sigma^m}$ converges a.s. to some bounded a.s. predictable process $\phi$, and $\xi^{\sigma^m}$ converges a.s. to some $\xi \in L^{\infty}_+(\Omega, \cF_t, \bP; \bR)$.
This gives us that $G(\phi^{\sigma^m}, \xi^{\sigma^m}, t, W)-X^{\sigma^m}$ converges a.s. to some nonnegative random variable. Therefore $\bK^a(t, W)-L^0_+(\Omega, \cF_T, \bP; \bR)$ is $\bP$-closed.

We now argue that there exists a risk-neutral measure for $\cK^a(t, W)$.
Towards this, we define the convex cone $\cC^a:= (\cK^a(t, W)-L^0_+(\Omega, \cF_T, \bP; \bR))\cap L^1(\Omega, \cF_T, \bP; \bR)$.
Due to the closedness property of $\bK^a(t, W)-L^0_+(\Omega, \cF_T, \bP; \bR)$,  we have that the set $\cC^a$ is closed in $ L^1(\Omega, \cF_T, \bP; \bR)$.
As in the proof of Theorem~\ref{Theorem: FFTAP}, we may construct a measure $\bQ \in \cZ$ such that $W$ is $\bQ$-integrable, and  $\bE_\bQ [K^a] \leq 0$ for all $K^a \in  \cK^a(t, X)$.
This completes the proof.
\end{proof}

\bibliographystyle{alpha}

\end{document}